\documentclass[smallextended,envcountsame]{svjour3}

\smartqed

\usepackage{tcolorbox}
\usepackage{amsmath}
\usepackage{amsfonts}
\usepackage{verbatim}
\usepackage{bbm}

\usepackage{todonotes}
\usepackage{cancel}
\usepackage{xspace}

\newcommand{\PPAD}{\ensuremath{\mathtt{PPAD}}\xspace}

\newcommand{\reals}{\mathbb{R}}

\newcommand{\eps}{\ensuremath{\epsilon}\xspace}
\newcommand{\supp}{\mathrm{supp}}
\newcommand{\suppmax}{\mathrm{suppmax}}
\newcommand{\mixedstrats}{\Delta}
\newcommand{\mixedstratslong}{\Delta_{m_1} \times \ldots \times \Delta_{m_n}}
\newcommand{\prof}{\ensuremath{\mathbf{x}}\xspace}
\newcommand{\profx}{\ensuremath{\mathbf{x}}\xspace}
\newcommand{\profy}{\ensuremath{\mathbf{y}}\xspace}
\newcommand{\profz}{\ensuremath{\mathbf{z}}\xspace}
\newcommand{\dual}{\ensuremath{\mathbf{q}}\xspace}
\newcommand{\degree}{\ensuremath{{d(i)}\xspace}}
\newcommand{\pmax}{\ensuremath{{U}\xspace}}
\newcommand{\pmin}{\ensuremath{{L}\xspace}}

\newcommand{\neighbours}{\ensuremath{N(i)}\xspace}

\newcommand{\brp}[2]{\ensuremath{u_#2^*(#1)}\xspace}
\newcommand{\brs}{\ensuremath{{\mathrm{Br}_{i}(\prof)}\xspace}}
\newcommand{\brst}{\ensuremath{{\mathrm{Br^{\delta}_{i}}(\prof^*)}\xspace}}
\newcommand{\brd}{\ensuremath{{\mathrm{Br^{\delta}_{i}}(\prof)}\xspace}}

\newcommand{\pays}[2]{\ensuremath{v_#2(#1)}\xspace}
\newcommand{\pay}[2]{\ensuremath{u_#2(#1)}\xspace}
\newcommand{\payt}[3]{\ensuremath{u_#3(#1,#2)}\xspace}
\newcommand{\maxed}[3]{\ensuremath{M_#3(#1,#2)}\xspace}
\begin{document}

\title{Computing Approximate Nash Equilibria in Polymatrix Games
\thanks{
The first 
author is supported by the Microsoft Research PhD sponsorship program.
The second and third authors are supported by EPSRC grant EP/L011018/1, and
the third author is also supported by ESRC grant ESRC/BSB/09.
The work of the fourth author is supported partially by the EU ERC Project ALGAME and by the Greek THALIS action ``Algorithmic Game Theory''.
%A full version of this paper is available at \texttt{http://arxiv.org/abs/1204.0707}
}
}
\author{
Argyrios Deligkas%\inst{1}
\and
John Fearnley%\inst{1}
\and 
Rahul Savani%\inst{1}
\and 
Paul Spirakis%\inst{1,2}
%\thanks{Supported by EPSRC grant ``Efficient Decentralised Approaches in Algorithmic Game Theory''}
}
\institute{A. Deligkas \and J. Fearnley \and R. Savani \and P. Spirakis \at Department of Computer Science, University of Liverpool, UK \and
P. Spirakis \at
Research Academic Computer Technology Institute (CTI), Greece}

\maketitle

\begin{abstract}
In an $\epsilon$-Nash equilibrium, a player can gain at most $\epsilon$ by
unilaterally changing his behaviour. For two-player (bimatrix) games with
payoffs in $[0,1]$, the best-known~$\epsilon$ achievable in polynomial time is
0.3393~\cite{TS}. In general, for $n$-player games an $\epsilon$-Nash
equilibrium can be computed in polynomial time for an $\epsilon$ that is an
increasing function of $n$ but does not depend on the number of strategies of
the players. For three-player and four-player games the corresponding values of
$\eps$ are 0.6022 and 0.7153, respectively. Polymatrix games are a restriction
of general $n$-player games where a
player's payoff is the sum of payoffs from a number of bimatrix games. There
exists a very small but constant $\eps$ such that computing an $\eps$-Nash
equilibrium of a polymatrix game is \PPAD-hard. Our main result is that a
$(0.5+\delta)$-Nash equilibrium of an $n$-player polymatrix game can be computed
in time polynomial in the input size and $\frac{1}{\delta}$. Inspired by the
algorithm of Tsaknakis and Spirakis~\cite{TS}, our algorithm uses gradient
descent on the maximum regret of the players. We also show that this algorithm
can be applied to efficiently find a $(0.5+\delta)$-Nash equilibrium in a
two-player Bayesian game.

\end{abstract}

\section{Introduction}

\paragraph{\bf Approximate Nash equilibria.}
Nash equilibria are the central solution concept in game theory. 
%is a strategy profile in which all players only assign
%probability to best responses. 
%
Since it is known that computing an \emph{exact} Nash equilibrium~\cite{DGP,CDT}
is unlikely to be achievable in polynomial time, a line of work has arisen that
studies the computational aspects of approximate Nash equilibria.
%, and two notions of approximate Nash equilibria have been developed. 
%
The most widely studied notion is of an \emph{$\epsilon$-approximate Nash
equilibrium} ($\epsilon$-Nash), which requires that all players have an expected
payoff that is within~$\epsilon$ of a best response. 
This is an \emph{additive} notion of approximate equilibrium; the problem of
computing approximate equilibria of bimatrix games using a relative notion of
approximation is known to be 
\PPAD-hard even for constant approximations~\cite{Das13}.

So far, $\epsilon$-Nash equilibria have mainly been studied in the context of
two-player \emph{bimatrix} games.
A line of work~\cite{DMP,Progress,BBM10} has investigated the best $\epsilon$
that can be guaranteed in polynomial time for bimatrix games. The current best
result, due to Tsaknakis and Spirakis~\cite{TS}, is a polynomial-time algorithm that
finds a 0.3393-Nash equilibrium of a bimatrix game with all payoffs in $[0,1]$.

% Now multi-player
In this paper, we study $\epsilon$-Nash equilibria in the context of
\emph{many-player} games, a topic that has received much less attention. A
simple approximation algorithm for many-player games can be obtained by
generalising the algorithm of Daskalakis,  Mehta and Papadimitriou~\cite{DMP}
from the two-player setting to the $n$-player setting, which provides a
guarantee of~$\eps = 1-\frac{1}{n}$.
This has since been improved independently by three sets of
authors~\cite{BGR08,HRS08,BBM10}. They provide a method that converts a
polynomial-time algorithm that for finding $\epsilon$-Nash equilibria in
$(n-1)$-player games into an algorithm that finds a $\frac{1}{2-\epsilon}$-Nash
equilibrium in $n$-player games. Using the polynomial-time $0.3393$ algorithm of
Tsaknakis and Spirakis~\cite{TS} for $2$-player games as the base case for this
recursion, this allows us to provide polynomial-time algorithms with
approximation guarantees of $0.6022$ in $3$-player games, and $0.7153$ in
$4$-player games. These guarantees tend to $1$ as
$n$ increases, and so far, no constant $\epsilon<1$ is known such that, for
all~$n$, an $\epsilon$-Nash equilibrium of an $n$-player game can be computed in
polynomial time.

For $n$-player games, we have lower bounds for $\epsilon$-Nash equilibria. More
precisely, Rubinstein has shown that when $n$ is not a constant there exists a
constant but very small~$\eps$ such that it is \PPAD-hard to compute an
$\eps$-Nash equilibrium~\cite{Rub14b}. This is quite different from the
bimatrix game setting, where the existence of a quasi-polynomial time
approximation scheme rules out such a lower bound, unless all of \PPAD can be
solved in quasi-polynomial time~\cite{LMM03}.

\paragraph{\bf Polymatrix games. } 
In this paper, we focus on a particular class of many-player games called
\emph{polymatrix games.} In a polymatrix game, the interaction between the
players is specified by an $n$ vertex graph, where each vertex represents one of
the players. Each edge of the graph specifies a bimatrix game that will be
played by the two respective players, and thus a player with degree $d$ will
play $d$ bimatrix games simultaneously. More precisely, each player picks a
strategy, and then plays this strategy in \emph{all} of the bimatrix games that
he is involved in. His payoff is then the sum of the payoffs that he obtains in
each of the games.

%A graphical game models the interactions between players via a game graph, 
%where players are vertices, and there is an edge between two players if the
%actions of one affects the payoff of the other.
%A strategic-form game corresponds to a graphical game with a complete game
%graph.
%When the maximum degree of the game graph is a constant, but the number of 
%players $n$ is not, then the graphical game representation is exponentially
%smaller than the strategic form.
%Polymatrix games are a restriction of graphical games in which 
%there is a bimatrix game associated with every edge.
%Each player chooses a single
%strategy that is used in all the bimatrix games that he plays with his
%neighbors in the game graph, with a resulting overall payoff that is the sum 
%of the payoffs from these bimatrix games.
%
Polymatrix games are a class of \emph{succinctly represented} $n$-player games:
a polymatrix game is specified by at most $n^2$ bimatrix games, each of which
can be written down in quadratic space with respect to the number of strategies.
This is unlike general $n$-player strategic form games, which require a
representation that is exponential in the number of players.

There has been relatively little work on approximation algorithms for
polymatrix games. The approximation algorithms for general games can 
be applied in this setting in an obvious way, but to the best of our knowledge there have been no
upper bounds that are specific to polymatrix games. On the other hand, the
lower bound of Rubinstein mentioned above is actually proved by constructing
polymatrix games. Thus, there is a constant but very small~$\eps$ such that it
is \PPAD-hard to compute an $\eps$-Nash equilibrium~\cite{Rub14b}, and this
again indicates that approximating polymatrix games is quite different to
approximating bimatrix games.

\paragraph{\bf Our contribution.}

Our main result is an algorithm that, for every $\delta$ in the range $0 <
\delta \leq 0.5$, finds a $(0.5 + \delta)$-Nash equilibrium of a polymatrix game
in 
time polynomial in the input size and
$\frac{1}{\delta}$.
Note that our approximation guarantee \emph{does not depend
on the number of players,} which is a property that was not previously known to
be achievable for polymatrix games, and still cannot be achieved for general
strategic form games.

We prove this result by adapting the algorithm of Tsaknakis and
Spirakis~\cite{TS} (henceforth referred to as the TS algorithm). They give a gradient descent algorithm for finding a
$0.3393$-Nash equilibrium in a bimatrix game. We generalise their gradient
descent techniques to the polymatrix setting, and show that it always arrives
at a $(0.5 + \delta)$-Nash equilibrium after a polynomial number of iterations.

%In this paper, we consider $n$-player polymatrix games.
%Such a game is defined by an undirected game graph, where 
%every vertex corresponds to a player.
%Every such player $i$ can have an arbitrary degree $d(i)$, which is
%the number of edges to some subset of the other
%$n-1$ players.
%Every player $i$ has a fixed number of pure strategies $m_i$.
%Every edge $\{i,j\}$ corresponds to a bimatrix game of size $m_i \times m_j$
%with payoffs in $[0,1]$, so the maximum possible payoff of player $i$ is $d(i)$.
%Our contribution is an algorithm that runs in polynomial time (in $n$ and the
%total number of pure strategies) and computes an $0.5$-Nash equilibrium.
%Our approximation guarantee is that every player $i$'s expected payoff is
%within $0.5\cdot d(i)$ of player $i$'s best response payoff, which gives 
%the first non-trivial approximation guarantee in this setting that does not
%depend on~$n$.

%Sometimes polymatrix games are modelled with complete game graphs, even 
%when some players do not interact, by placing 
%bimatrix games with all zero payoffs on the ``missing edges''.
%In that case, where we essentially ignore players' individual maximum payoffs, 
%an approximation guarantee of $\eps$ would require players 
%payoffs to be within $\eps\cdot(n-1)$ of a best response; our result 
%is stronger,
%with a guarantee to be within $\eps\cdot d(i)$ of a best response for any 
%game graph.

In order to generalise the TS algorithm, we had to overcome several issues.
Firstly, the TS algorithm makes the regrets of the two players equal in every
iteration, but there is no obvious way to achieve this in the polymatrix setting.
Instead, we show how gradient descent can be applied to a strategy profile
where the regrets are not necessarily equal. Secondly, the output of the TS
algorithm is either a point found by gradient descent, or a point obtained by
modifying the result of gradient descent. In the polymatrix game setting, it is
not immediately obvious how such a modification can be derived with a
non-constant number of players (without an exponential blowup). Thus we apply a
different analysis, which proves that the point resulting from gradient descent
always has our approximation guarantee. It is an interesting open
question whether a better approximation guarantee can be achieved when there is
a constant number of players.

An interesting feature of our algorithm is that it can be applied even when
players have differing degrees. Originally, polymatrix games were defined only
for complete graphs~\cite{H72}. Since previous work has only considered lower
bounds for polymatrix games, it has been sufficient to restrict attention to 
regular graphs, as in work Rubinstein~\cite{Rub14b}.
However, since this paper is proving an
upper bound, we must be more careful. As it turns out, our algorithm will
efficiently find a
$(0.5+\delta)$-Nash equilibrium for all $\delta>0$, no matter what
graph structure the polymatrix game has.

Finally, we show that our algorithm can be applied to two-player Bayesian games.
In a two-player Bayesian game, each player is assigned a type according to a
publicly known probability distribution. Each player knows their own type, but
does not know the type of their opponent. We show that finding an
$\epsilon$-Nash equilibrium in these games can be reduced to the problem of
finding an $\epsilon$-Nash equilibrium in a polymatrix game, and therefore, our
algorithm can be used to efficiently find a $(0.5+\delta)$-Nash equilibrium of
a two-player Bayesian game.

\paragraph{\bf Related work.}

% PTAS 
An FPTAS for the problem of computing an $\epsilon$-Nash equilibrium of a
bimatrix game does not exist unless every problem in \PPAD can be solved in
polynomial time~\cite{CDT}. Arguably, the biggest open question in equilibrium
computation is whether there exists a PTAS for this problem.
As we have mentioned, for any constant $\eps>0$, there does exist a
\emph{quasi-polynomial}-time algorithm for computing an $\eps$-Nash equilibria
of a bimatrix game, or any game with a constant number of
players~\cite{LMM03,BBP14}, with
running time $k^{O(\log k)}$ for a $k \times k$ bimatrix game.
Consequently, in contrast to the many-player case, it is not believed that there
exists a constant $\eps$ such that the problem of computing an $\epsilon$-Nash
equilibrium of a bimatrix game (or any game with a constant number of players)
is \PPAD-hard, since it seems unlikely that all problems in \PPAD have
quasi-polynomial-time algorithms.
On the other hand, for multi-player games, as mentioned above, there is a small
constant $\eps$ such that it is \PPAD-hard to compute an \eps-Nash equilibrium
of an $n$-player game when $n$ is not constant.
One positive result we do have for multi-player games is that there is a PTAS 
for \emph{anonymous games} (where the identity of players does not matter)
when the number of strategies is constant~\cite{DP14}.

Polymatrix games have played a central role in the reductions that have 
been used to show \PPAD-hardness of games and other equilibrium 
problems~\cite{DGP,CDT,EY10,FT-C10,CPY13}.
Computing an \emph{exact} Nash equilibrium in a polymatrix game is
\PPAD-hard even when all the bimatrix games played are either zero-sum
games or coordination games~\cite{CD11}.
%
%It was shown (though not stated explicitly) in~\cite{DGP} that the problem of finding 
%an exact Nash equilibrium of a polymatrix game with two strategies per player
%is \PPAD-hard.
%Recently it was shown that it is PPAD-hard to find an
%$\epsilon$-WSNE of an $n$-player polymatrix game with 
%$\epsilon = 1/n$~\cite{CPY13}.
%
%%%%%%%%%%%%%%%%%%%%%%%%%%%%%%%%%%%%%%%%%%%%%%%%%%%%%%%%%%%%%%%%%%
% Govindan and Wilson
%%%%%%%%%%%%%%%%%%%%%%%%%%%%%%%%%%%%%%%%%%%%%%%%%%%%%%%%%%%%%%%%%%
%
%
Polymatrix games have been used in other contexts too.
For example, Govindan and Wilson proposed a (non-polynomial-time) algorithm for computing
Nash equilibria of an $n$-player game, by approximating the game with a
sequence of polymatrix games~\cite{GW04}.  Later, they presented a
(non-polynomial) reduction that reduces $n$-player games to polymatrix games
while preserving approximate Nash equilibria \cite{GW10}.  Their reduction
introduces a central coordinator player, who interacts bilaterally with every
player.
% while simulating the combined behavior of the other players.

%\paragraph{\bf Road map.} DO WE WANT ONE? JOHN VOTES NO!

\section{Preliminaries}
\label{sec:prelim}

%In this section, we define an $n$-player polymatrix game and related notions.
We start by fixing some notation. We use $[k]$ to denote the set of integers
$\{1, 2, \ldots, k\}$, and when a universe $[k]$ is clear, we will use $\bar{S}
= \{ i \in [k], i \notin S\}$ to denote the complement of $S \subseteq [k]$.
For a $k$-dimensional vector $x$, we use $x_{-S}$ to denote the elements of $x$
with with indices~$\bar{S}$, and in the case where $S = \{i\}$ has only one
element, we simply write $x_{-i}$ for $x_{-S}$.

\paragraph{\bf Polymatrix games.} An $n$-player polymatrix game is defined by an
undirected graph $(V, E)$ with $n$ vertices, where every vertex corresponds to a
player. The edges of the graph specify which players interact with each other.
For each $i \in [n]$, we use $\neighbours = \{j \; : \; (i,j) \in E\}$ to
denote the neighbours of player $i$.
%, and we use $\degree = |\neighbours|$ to
%denote player $i$'s \emph{degree}. We also define $\maxdeg = \max_{i \in [n]}
%\degree$ to be the maximum degree of the game.
%

Each edge $(i, j) \in E$ specifies that a bimatrix game will be played between
players $i$ and $j$. Each player $i \in [n]$ has a fixed number of pure
strategies $m_i$, and the bimatrix game on edge $(i, j) \in E$ will therefore be
specified by an $m_i \times m_j$ matrix $A_{ij}$, which gives the payoffs for
player $i$, and an $m_j \times m_i$ matrix $A_{ji}$, which gives the payoffs for
player $j$. We allow the individual payoffs in each matrix to be an arbitrary
(even negative) rational number. As we describe in the next subsection, we will
rescale these payoffs so that the overall payoff to each player lies in the
range $[0, 1]$. 
%We assume that all payoff lie in the range $[0, 1]$, and therefore each payoff
%matrix $A_{ij}$ lies in $[0,1]^{m_i \times m_j}$.
%\todo[inline]{No we don't}

\subsection{Payoff Normalization} 

%\begin{itemize}
%\item Some assumption on the size of payoffs is required so that there is a consistent
	%meaning of $\epsilon$ across games
%\item The polymatrix games where payoffs arise as sums of payoffs in individual
	%bimatrix games - one approach is to restrict the payoffs in those
	%individual games
%\item If all bimatrix in [0,1] then max payoff is degree, and WINE paper
	%divided by degree
%\item Example where stronger definition bites (when different bimatrix games
	%have different payoffs = significance to overall payoff) 
%\end{itemize}

Before we continue, we must first discuss how the payoffs in the game are
rescaled. It is common, when proving results about additive notions of
approximate equilibria, to rescale the payoffs of the game. This is necessary in
order for different results to be comparable. For example, all results about
additive approximate equilibria in bimatrix games assume that the payoff
matrices have entries in the range $[0, 1]$, and therefore an $\epsilon$-Nash
equilibrium always has a consistent meaning. For the same reason, we must
rescale the payoffs in a polymatrix in order to give a consistent meaning to
an $\epsilon$-approximation.

%In this subsection, we discuss the
%various ways in which polymatrix games can be rescaled.

An initial, naive, approach would be to specify that each of the individual
bimatrix games has entries in the range $[0, 1]$. This would be sufficient if we
were only interested in polymatrix games played on either complete graphs or
regular graphs. However, in this model, if the players have differing degrees,
then they also have differing maximum payoffs. This means that an additive
approximate equilibrium must pay more attention to high degree players, as they
can have larger regrets.

One solution to this problem, which was adopted in the conference version of
this paper~\cite{DFSS14}, is to rescale according to the degree. That is, given
a polymatrix game where each bimatrix game has payoffs in the range $[0, 1]$, if
a player has degree $d$, then each of his payoff matrices is divided by $d$.
This transformation ensures that every player has regret in the range $[0, 1]$,
and therefore low degree players are not unfairly treated by additive
approximations.

However, rescaling according to the degree assumes that each bimatrix game
actually uses the full range of payoffs between $[0, 1]$. In particular, some
bimatrix games may have minimum payoff strictly greater than $0$, or maximum
payoff strictly less than $1$. This issue arises, in particular, in our
application of two-player Bayesian games. Note that, unlike the case of
a single bimatrix game,
we cannot fix this by rescaling individual bimatrix games in a polymatrix game,
because we must maintain the relationship between the payoffs in all of the
bimatrix games that a player is involved in.

%Note that we do not make any further assumptions
%for the entries of payoff matrices $B_{ij}$. This mean that the maximum payoff 
%for each player $i$ is bounded by its degree $\degree$. In order to be able to
%define the correct notion of approximate equilibrium we must scale the total
%payoff for each player (that is the payoff he gets over all the bimatrix games 
%he plays) to $[0, 1]$. Next we describe how we can achieve this.

To address this, we will rescale the games so that, for each player, the minimum
possible payoff is $0$, and the maximum possible payoff is $1$. For each
player~$i$, we denote by $\pmax$ the maximum payoff he can obtain, and by $\pmin$ the
minimum payoff he can obtain. Formally:
\begin{align*}
\pmax_i & := \max_{p \in [m_i]} \left( \sum_{j \in \neighbours} \max_{q \in [m_j]} \big( A_{ij}(p,q) \big) \right),\\
\pmin_i & := \min_{p \in [m_i]} \left( \sum_{j \in \neighbours} \min_{q \in [m_j]} \big( A_{ij}(p,q) \big) \right).
\end{align*}
Then, for all $i$ and all $j \in \neighbours$ we will apply the following
transformation, which we call $T(\cdot)$, to all the entries $z$ of payoff
matrices $A_{ij}$:
\begin{align*}
T_i(z) = \frac{1}{\pmax_i - \pmin_i} \cdot \left( z - \frac{\pmin_i}{\degree} \right).
\end{align*}
%
%\todo[inline]{Do we need to argue that this transformation indeed produces total payoffs in $[0,1]$?}
Observe that, since player $i$'s payoff is the sum of $\degree$ many bimatrix
games, it must be the case that after transforming the payoff matrices in this
way, player $i$'s maximum possible payoff is $1$, and player $i$'s minimum
possible payoff is $0$. For the rest of this paper, we will assume that the
payoff matrices given by $A_{ij}$ are rescaled in this way.

%It is worth noting that this rescaling gives stronger results than degree-based
%rescaling: every $\epsilon$-NE under this rescaling is also an $\epsilon$-NE
%under the degree-based rescaling, but the converse does not hold. Thus, this
%rescaling can only make it more difficult to find a $0.5$-NE, which means our
%overall result is stronger.

%We will denote the transformed payoff matrices by $A_{ij}$. 
%After this transformation the minimum payoff that a player can get is zero 
%and the  maximum one. Note that some entries of matrices $A_{ij}$ can be 
%negative, but this does not create a problem since the algorithm 
%considers only the total payoff for each player. In what follows we use only 
%$A_{ij}$ as payoff matrices, since the Nash equilibria of the game does not change
%under transformation $T(\cdot)$.

\subsection{Approximate Nash Equilibria}

\paragraph{\bf Strategies.} 

A \emph{mixed strategy} for player $i$ is a
probability distribution over player $i$'s pure strategies. Formally, for each
positive integer $k$, we denote the $(k-1)$-dimensional simplex by $\Delta_k :=
\{x: x \in \reals^k, x \geq 0, \sum_{i=1}^k x_i = 1\}$, and therefore the set of
strategies for player $i$ is $\Delta_{m_i}$.
For each mixed strategy $x \in \Delta_m$, the \emph{support} of $x$ is defined
as $\supp(x) := \{i \in [m]: x_i \neq 0 \}$, which is the set of strategies
played with positive probability by $x$. 

A \emph{strategy profile} specifies a mixed strategy for every player. We denote
the set of mixed strategy profiles as $\mixedstrats := \mixedstratslong$. Given
a strategy profile $\mathbf{x} = (x_1,\ldots,x_n) \in \mixedstrats$, the payoff
of player $i$ under $\mathbf{x}$ is the sum of the payoffs that he obtains in
each of the bimatrix games that he plays. Formally, we define: 
\begin{align}
\label{payoff}
\pay{\prof}{i} := x_i^T\sum\limits_{j \in \neighbours}A_{ij}x_j.
\end{align}
We denote by $\payt{x'_i}{\prof}{i}$ the payoff for player $i$ when he plays
$x'_i$ and 
the other players play according to the strategy profile $\prof$. 
%This will be useful because in many cases we need to compute the payoff for a
%player when he plays a fixed strategy and the other players play  deviate, 
%A typical example of use is 
%$\payt{x_i}{\prof'}{i}$, which is the payoff for player $i$ when he stays in
%his current strategy $x_i$ from the profile $\prof$ and the other players play
%according to the profile~$\prof'$. 
In some cases the first argument will be $x_i - x'_i$ which may not correspond to a valid strategy for player 
$i$ but we still apply the equation as follows:
\begin{equation*}
%\label{eq:payoff-ext}
\payt{x_i-x'_i}{\prof}{i} := 
x_i^T\sum\limits_{j \in \neighbours}A_{ij}x_j - x_i'^T\sum\limits_{j \in \neighbours}A_{ij}x_j = 
\payt{x_i}{\prof}{i} - \payt{x'_i}{\prof}{i}.
\end{equation*}

\paragraph{\bf Best responses.}
Let $\pays{\prof}{i}$ be the vector of payoffs for each pure strategy of player
$i$ when the rest of players play strategy profile $\prof$. Formally,
\begin{equation*}
%\label{eq:pays}
% S_{A_i}(x_{-i}) = \left \{ k \in [m_i], \left ( \sum \limits_{j \neq i} A_{ij}x_j \right )_k \geq \left ( \sum \limits_{j \neq i} A_{ij}x_j \right )_{k'}, \forall k' \right \}
\pays{\prof}{i} = \sum_{j \in \neighbours} A_{ij}x_j.
\end{equation*}
For each vector $x \in R^m$, we define $\suppmax(x)$ to be the set of indices
that achieve the maximum of $x$, that is, we define $\suppmax(x) = \{ i \in [m]:
x_i \geq x_j, \forall j\in[m] \}$. Then the \emph{pure best responses} of player
$i$ against a strategy profile \prof (where only $\prof_{-i}$ is relevant) is
given by:
\begin{equation}
\label{eq:brs}
% S_{A_i}(x_{-i}) = \left \{ k \in [m_i], \left ( \sum \limits_{j \neq i} A_{ij}x_j \right )_k \geq \left ( \sum \limits_{j \neq i} A_{ij}x_j \right )_{k'}, \forall k' \right \}
\brs = \suppmax\left(\sum_{j \in \neighbours} A_{ij}x_j\right ) = \suppmax(\pays{\prof}{i}).
\end{equation}
The corresponding \emph{best response payoff} is given by:
\begin{equation}
\label{eq:brp}
\brp{\prof}{i} = \max_k \left \{ \bigl (\sum_{j \in \neighbours} A_{ij}x_j\bigr )_k \right \} = \max_k \left \{ \bigl ( \pays{\prof}{i} \bigr )_k \right\}.
\end{equation}

\paragraph{\bf Equilibria.}

%Indeed, if we were only interested in
%exact equilibria, then the same definition could also be used here. However,
%when we consider additive approximate equilibria, as we do in this paper, we
%must be more careful.

%Recall that the bimatrix games on the edges of the graph have payoffs in the
%range $[0, 1]$. Therefore, the maximum possible payoff that player $i$ can
%receive is~$\degree$, which then implies that the maximum possible regret
%suffered by player $i$ is $\degree$. This causes a problem for low degree
%players: if a player has low degree, then their regret is always close to their
%best response payoff, and thus an additive approximation has little incentive to
%care about them, when compared to a high degree player who can have a very large
%regret.

%For every player $i \in [n]$ and strategy profile $\prof \in \mixedstrats$, the
%\emph{regret} for player $i$ under \prof is the amount that $i$ can gain by a
%unilateral deviation.
%
%We will perform gradient descent on a function that compute the maximum regret over the players.
%
%The regret and maximum regret functions, as well as an \eps-Nash equilibrium
%are defined formally as follows.
%
%Note that the payoffs in every individual bimatrix game lie in the range
%$[0,1]$, so the overall payoffs for player $i \in [n]$ lie in the range 
%$[0,\degree]$, which accounts for the scaling of each player's regret by $\degree$ 
%in the following definition, so that 
%our results are comparable with the literature, where overall payoffs are in
%$[0,1]$.

In order to define the exact and approximate equilibria of a polymatrix game,
we first define the \emph{regret} that is suffered by each player under a given
strategy profile. The regret function  $f_i:\mixedstrats \rightarrow [0, 1]$ is 
defined, for each player $i$, as follows: 
\begin{align} 
\label{eq:regret} 
f_i(\prof) := \brp{\prof}{i} - \pay{\prof}{i}.
\end{align}
%Note that here we have taken the standard definition of regret, and divided
%it by the player's maximum payoff minus his minimum payoff. 
%\todo{changes}
%This rescales the regret so that it lies in the range
%$[0, 1]$, which ensures that our approximation guarantee is comparable with
%those in the literature. 
The maximum regret under a strategy profile $\prof$ is given by the
function $f(\prof)$ where:
\begin{align}
\label{eq:f}
f(\prof) := \max\{f_1(\prof), \ldots, f_n(\prof)\}.
\end{align}
We say that $\prof$ is an $\epsilon$-approximate Nash equilibrium
($\epsilon$-NE) if we have:
\begin{equation*}
%\label{epsNE}
f(\prof) \le \epsilon,
\end{equation*}
and $\prof$ is an \emph{exact} Nash equilibrium if we have
$f(\prof) = 0$.
%\begin{definition}
%\label{def:epsNE}

%%
%An $\epsilon$-approximate Nash equilibrium ($\epsilon$-NE) is 
%a strategy profile $\hat{\prof} \in \mixedstrats$
%with 
%%
%\end{definition}

%Clearly, for all $\prof \in \mixedstrats$ and $i \in [n]$, 
%we have $0 \leq f_i(\prof) \leq 1$ and thus $0 \leq f(\prof) \leq 1$,
%and we have that $f(\prof) = 0$ if and only if \prof is an exact Nash equilibrium.
%%
%In the next section, we define a gradient for the maximum regret function $f(\prof)$.

%We argue that this is a reasonable notion of approximation for polymatrix
%games where players have differing degrees. Previously, polymatrix games have
%either been considered on complete graphs, or in the case of Rubinstein's
%work~\cite{Rub14b}, for regular graphs. In these cases, where the players all
%have the same degree, our rescaling affects all players equally, and our notion
%of approximation is the natural one. In the case where players have differing
%degrees, our notion places additional constraints, by requiring that players
%with differing degrees are treated equally. 

%Note that, under this definition, an $\epsilon$-NE treats every player equally:
%we are only in an $\epsilon$-NE if the payoff obtained every player $i$ is
%within $\frac{\epsilon}{\degree}$ of their best response payoffs.

%%%%%%%%%%%%%%%
%%%%%%%%%%%%%%%
%%%%%%%%%%%%%%%

\section{The gradient}

Our goal is to apply gradient descent to the regret function $f$. In this
section, we formally define the gradient of $f$ in
Definition~\ref{def:gradient}, and give a reformulation of that definition in
Lemma~\ref{gradient-lem}. In order to show that our gradient descent method
terminates after a polynomial number of iterations, we actually need to use a
slightly modified version of this reformulation, which we describe at the end of
this section in Definition~\ref{def:dgrad}.

%First we formally define the gradient of $f$. 
Given a point $\prof \in
\mixedstrats$, a \emph{feasible direction} from $\prof$ is defined by any other
point $\prof' \in \mixedstrats$. This defines a line between $\prof$ and
$\prof'$, and formally speaking, the direction of this line is $\prof' - \prof$.
%
%From any point $\prof \in \mixedstrats$, we consider changes to $f(\prof)$.
%
In order to define the gradient of this direction, we consider the function
$f((1-\eps)\cdot \prof + \eps \cdot \prof') - f(\prof)$ where $\epsilon$ lies in
the range $0 \leq \epsilon \leq 1$. The gradient of this direction is given in
the following definition.
\begin{definition}
\label{def:gradient}
Given profiles $\prof,\prof' \in \mixedstrats$ and $\epsilon \in [0,1]$, we
define:
\begin{align*}
% Df(x_1, \ldots, x_n, x'_1, \ldots, x'_n, \epsilon) = f(x_1 + \epsilon(x_1' - x_1), \ldots, x_n + \epsilon(x'_n - x_n)) - f(x_1, \ldots, x_n)
Df(\prof,\prof', \epsilon) & := f((1 - \eps) \cdot \prof + \epsilon \cdot
\prof') - f(\prof).
%\\
% & = f(x_1 + \epsilon(x_1' - x_1), \ldots, x_n + \epsilon(x'_n - x_n)) - f(x_1, \ldots, x_n).
\end{align*}
Then, we define the gradient of $f$ at $\prof$ in the direction $\prof'-\prof$ as:
\begin{equation*}
%\label{gradient-def}
Df(\prof,\prof') = \lim \limits_{\epsilon \rightarrow 0}\frac{1}{\epsilon}Df(\prof,\prof', \epsilon).
\end{equation*}
\end{definition}
This is the natural definition of the gradient, but it cannot be used directly
in a gradient descent algorithm. 
We now show how this definition can be reformulated. Firstly, for each $\prof,
\prof' \in \mixedstrats$, and for each player $i \in [n]$, we define:
\begin{equation}
\label{eq:Df_i}
Df_i(\prof, \prof') := \max_{k \in \brs} \left \{ \bigl ( \pays{\prof'}{i}
\bigr )_k \right \} - \payt{x_i}{\prof'}{i} + \payt{x_i-x'_i}{\prof}{i} .
\end{equation}
Next we define $\mathcal{K}(\prof)$ to be the set of players that have maximum
regret under the strategy profile~\prof.
\begin{definition}
Given a strategy profile \prof, define $\mathcal{K}(\prof)$ as follows:
\begin{equation}
\label{eq:k-set}
\mathcal{K}(\prof) := \bigl\{ i \in [n], f_i(\prof) = f(\prof) \bigr\} = 
\left\{ i \in [n], f_i(\prof) = \max_{j \in [n]} f_j(\prof)
\right\}.
\end{equation}
\end{definition}

The following lemma, which is proved in Appendix~\ref{app:gradient-lem},
provides our reformulation. 
%
%Next we show an alternative formulation for the gradient that will be useful
%for proving that only polynomially many iterations of gradient descent are needed in order to find
%a $\prof$ such that $f(\prof) \leq \frac{1}{2}$, and for proving 
%bounds the quality of approximation of the resulting stationary points of $f$.
%%
%This alternative formulation uses the following quantities that are derived from
%$f_i$ for each $i \in [n]$.
%For $\prof, \prof' \in \mixedstrats$, we define:
%\end{definition}
%Next we provide the alternative formulation of the gradient and
%prove that it is equivalent to Definition~\ref{def:gradient}.
%
%These quantities will help us to define the gradient for $f$ in a given point. 
\begin{lemma}
\label{gradient-lem}
%Let $\mathcal{K}(\prof)$ as in \eqref{eq:k-set}.
The gradient of $f$ at point $\prof$ along direction $\prof' - \prof$ is: 
\begin{equation*}
%\label{eq:grad-alter}
Df(\prof, \prof') = \max_{i \in \mathcal{K}(\prof)} Df_i (\prof, \prof') - f(\prof).
\end{equation*}
\end{lemma}

%It is worth noting that the correctness of our result does not depend on
%Lemma~\ref{gradient-lem}. That is, we could have simply defined the gradient to
%be the expression given in Equation~\eqref{eq:grad-alter} and the proof that our
%algorithm finds a $(0.5+\delta)$-Nash algorithm remains unchanged. We present
%Lemma~\ref{gradient-lem} in order to show that the gradient given in
%Equation~\eqref{eq:grad-alter} actually corresponds to an intuitive notion of
%gradient.

%Lemma~\ref{gradient-lem} gives 

In order to show that our gradient descent algorithm terminates
after a polynomial number of steps, we have to use a slight modification of 
the formula given in Lemma~\ref{gradient-lem}. More precisely, in the definition
of $Df_i(\prof, \prof')$, we need to take the maximum over the $\delta$-best
responses, rather than the best responses. 

We begin by providing the definition of the $\delta$-best responses.
\begin{definition}[$\delta$-best response]
\label{def:dbr}
Let $\prof \in \Delta$, and let $\delta \in (0, 0.5]$.
The $\delta$-best response set $\brd$ for player $i \in [n]$ is defined as:
\begin{equation*} 
\brd := \left\{ j \in [m_i] : \bigl( \pays{\prof}{i} \bigr)_j \geq \brp{\prof}{i} - \delta \right\}.
\end{equation*}
\end{definition}
% Observe that we have multiplied $\delta$ by $\fscale$, because $\fscale$ is the
% maximum possible payoff that player $i$ can obtain. 
We now define the function
$Df^{\delta}_i(\prof,\prof')$.

\begin{definition}
\label{def:dgrad}
Let $\prof, \prof' \in \Delta$, let $\eps \in [0, 1]$, and let $\delta \in (0,
0.5]$. We define $Df^{\delta}_i(\prof,\prof')$ as:
\begin{equation}
\label{eq:delta-dfi}
Df^{\delta}_i(\prof,\prof') :=  \max_{k \in \brd} \left \{ \bigl ( \pays{\prof'}{i}
\bigr )_k \right \} - \payt{x_i}{\prof'}{i} - \payt{x'_i}{\prof}{i} + \payt{x_i}{\prof}{i}.
\end{equation}
Furthermore, we define $Df^{\delta}(\prof, \prof')$ as:
\begin{equation}
\label{eq:dgrad}
Df^{\delta}(\prof, \prof') = \max_{i \in \mathcal{K}(\prof)} Df^{\delta}_i (\prof, \prof') - f(\prof).
\end{equation}
\end{definition}

Our algorithm works by performing gradient descent using the function
$Df^{\delta}$ as the gradient. Obviously, this is a different function to $Df$,
and so we are not actually performing gradient descent on the gradient of $f$.
It is important to note that all of our proofs are in terms of $Df^{\delta}$,
and so this does not affect the correctness of our algorithm. We proved
Lemma~\ref{gradient-lem} in order to explain where our definition of the gradient
comes from, but the correctness of our algorithm does not depend on the
correctness of Lemma~\ref{gradient-lem}.

%
%%%%%%%%%%%%%%%%%%%%%%%%%%%%%%%%%%%%%%
%%%%%%%%%%%%%%%%%%%%%%%%%%%%%%%%%%%%%%%
%

\section{The algorithm}
\label{sec:stationary}

In this section, we describe our algorithm for finding a $(0.5+\delta)$-Nash equilibrium
in a polymatrix game by gradient descent. In each iteration of the algorithm, 
 we must find the \emph{direction} of steepest
descent with respect to $Df^{\delta}$. 
We show that this task can be achieved by solving a linear
program, and we then use this LP to formally specify our algorithm.

\paragraph{\bf The direction of steepest descent.}

We show that the direction of steepest descent can be found by solving a linear
program.
%
%We begin by introducing a linear
%program that will be used to find the direction of steepest
%descent at $f(\prof)$.
%
%Given a strategy profile \prof, a feasible direction in the strategy space is 
%essentially the combination of deviations that players can make from their 
%current strategies. 
%
%Among all possible feasible directions we will choose the
%one with the steepest descent. 
%
Our goal is, for a given strategy profile $\prof$, to find another strategy
profile $\prof'$ so as to minimize the gradient $Df^{\delta}(\prof, \prof')$. 
Recall that $Df^{\delta}$ is defined in Equation~\eqref{eq:dgrad} to be:
\begin{equation*}
Df^{\delta}(\prof, \prof') = \max_{i \in \mathcal{K}(\prof)} Df^{\delta}_i (\prof, \prof') - f(\prof).
\end{equation*}
Note that the term $f(\prof)$ is a constant in this expression, because it is
the same for all directions $\prof'$. Thus, it is sufficient to formulate a
linear program in order to find the $\prof'$ that minimizes 
$\max_{i \in \mathcal{K}(\prof)} Df^{\delta}_i
(\prof, \prof')$.  Using the definition of $Df^{\delta}_i$ in
Equation~\eqref{eq:delta-dfi}, we can do this as follows.

%By
%definition, this is equivalent to finding an $\prof'$ that minimizes the maximum
%of $Df^{\delta}_i(\prof, \prof')$ over $i \in \mathcal{K}(\prof)$. %as defined
%in Equation~\eqref{eq:Df_i}.

%We show that this can be formulated as a linear program. 

%Recall from
%Equation~\eqref{eq:delta-dfi} that
%: $$Df^{\delta}_i(\prof, \prof') :=
%\frac{1}{\degree} \left(\max_{k \in \brd} \bigl\{ \bigl( \pays{\prof'}{i}
%\bigr)_k \bigr\} - \payt{x'_i}{\prof}{i} - \payt{x_i}{\prof'}{i} +
%\pay{\prof}{i} \right).$$
%
%We will introduce a variable $l_i$ to deal with the maximum
%$\max_{k \in \brd} \bigl\{ \bigl( \pays{\prof'}{i} \bigr)_k \bigr\} \leq l_i$
%for every player $i \in \mathcal{K}(\prof)$ and a variable $w$ 
%in order to deal with the maximum of $Df^{\delta}_i(\prof, \prof')$ for all $i \in \mathcal{K}(\prof)$. Formally:
%
\begin{definition}[Steepest descent linear program]
\label{def:lp}
Given a strategy profile $\prof$, the \emph{steepest descent linear program} is
defined as follows.
Find $\prof' \in \mixedstrats$, $l_1,l_2,\ldots,l_{|\mathcal{K}(\prof)|}$, and $w$ such that:
\begin{align*}
%\nonumber
\text{minimize} \qquad & w \\
%\label{eq:type1-constraint}
\text{subject to} \quad & \bigl( \pays{\prof'}{i} \bigr)_k \leq l_i  \hspace*{25mm} \forall k \in \brd, &\forall i \in \mathcal{K}(\prof) \\
%\label{eq:type2-constraint}
&  l_i - \payt{x_i}{\prof'}{i} - \payt{x'_i}{\prof}{i} +  \pay{\prof}{i}  \leq w  & \forall i \in \mathcal{K}(\prof)\\
%\nonumber
&  \prof'  \in \Delta.
\end{align*}
\end{definition}
The $l_i$ variables deal with the maximum in the term $\max_{k \in
\brd} \bigl\{ \bigl( \pays{\prof'}{i} \bigr)_k \bigr\}$, while the
variable $w$ is used to deal with the maximum over the functions
$Df^{\delta}_i$. Since the constraints of the linear program
correspond precisely to the
definition of $Df^{\delta}$, it is clear that, when we minimize $w$, 
the resulting $\prof'$ specifies the direction of steepest descent. For
each profile $\prof$, we define $Q(\prof)$ to be the direction $\prof'$ found by
the steepest descent LP for $\prof$.

Once we have found the direction of steepest descent, we then need to move in
that direction. More precisely, we fix a parameter $\eps = \frac{\delta}{\delta
+ 2}$ which is used to determine how far we move in the steepest descent
direction. We will show in Section~\ref{sec:converge} that this value of
$\epsilon$ leads to a polynomial bound on the running time of our algorithm.

\paragraph{\bf The algorithm.}
We can now formally describe our algorithm.
%Recall that $\maxdeg = \max_{i \in [n]} \degree$.
The algorithm takes a parameter $ \delta \in (0, 0.5]$, which will be used as a
tradeoff between running time and the quality of approximation. 

\bigskip

\begin{tcolorbox}[title=Algorithm 1]                                                      
    
\begin{enumerate}
\itemsep2mm
\item Choose an arbitrary strategy profile $\prof \in \mixedstrats$.
\item Solve the steepest descent linear program with input $\prof$ to obtain~$\prof'
= Q(\prof)$.
%
%
%This can be achieved in polynomial time by applying line search on $f$, because \prof and $\prof'$ are fixed.
%
\item  
%$\epsilon=\frac{\delta}{\delta+\maxdeg + 1}$ and 
Set $\prof:= \prof + \eps(\prof' - \prof)$, where
$\eps = \frac{\delta}{\delta + 2}$.
%\footnote{
 %We use $\eps=\frac{\delta}{\delta+\maxdeg + 1}$ in the algorithm definition, since in
 %Section \ref{sec:converge} we use this value of $\eps$ in order to upper
 %bound the time complexity of the algorithm.
 %This value for $\eps$ does not give always the maximum decrease for the value of $f(\prof)$.
 %When implementing the algorithm, we can optimize $\epsilon$ as described in Observation \ref{lem:opt-eps}.
 %}
\item If $f(\prof) \leq 0.5 + \delta$ then stop, otherwise go to step 2.
\end{enumerate}

\end{tcolorbox}

\bigskip

A single iteration of this algorithm corresponds to executing steps 2, 3, and 4.
Since this only involves solving a single linear programs, it is clear that each
iteration can be completed in polynomial time.

The rest of this paper is dedicated to showing the following theorem, which is
our main result.
\begin{theorem}
\label{thm:main}
Algorithm 1 finds a $(0.5+\delta)$-NE after at most $O(\frac{1}{\delta^2})$ iterations.
\end{theorem}
To prove Theorem~\ref{thm:main}, we will show two properties. Firstly, in
Section~\ref{sec:stationary2}, we show that our gradient descent algorithm
never gets stuck in a stationary point before it finds a $(0.5 + \delta)$-NE. To
do so, we define the notion of a \emph{$\delta$-stationary point}, and we show
that every $\delta$-stationary point is at least a $(0.5 + \delta)$-NE, which
then directly implies that the gradient descent algorithm will not get stuck
before it finds a $(0.5 + \delta)$-NE.

Secondly, in Section~\ref{sec:converge}, we prove the upper bound on the number of
iterations. To do this we show that, if an iteration of the algorithm starts at
a point that is not a $\delta$-stationary point, then that iteration will make a
large enough amount of progress. This then allows us to show that the algorithm
will find a $(0.5 + \delta)$-NE after $O(\frac{1}{\delta^2})$ many
iterations, and therefore the overall running time of the algorithm is
polynomial.

%In order to prove that this algorithm is correct, we must prove two 
%Thus, in order to prove that the algorithm has polynomial running time, 
%we only need to prove that the algorithm makes
%we only need to prove that the steepest descent procedure 
%requires a polynomial number of iterations in order to find a point $\prof$ that
%$f(\prof) \leq 0.5$. 
%%
%Indeed this is the case, as we prove in the next section. 
%%
%More specifically, we prove that after
%%either converges in a $\delta-$stationary point or after 
%$O(\frac{\maxdeg^2}{\delta^2})$ iterations of the steepest descent procedure the
%algorithm either finds a point $\prof$ that
%$f(\prof) \leq 0.5$, or it converges to a $\delta$-stationary point.
%We should note that the initial strategy $\prof \in \mixedstrats$ chosen in 
%the first step of the algorithm affects the the quality of the approximation
%of the steepest descent procedure.

%%%%%%%%%%%%%%%%%%%%%%%%%%%%%%%%%%%%%%%%%%%%%%%%%%%%%%%%
%
%
%
%%%%%%%%%%%%%%%%%%%%%%%%%%%%%%%%%%%%%%%%%%%%%%%%%%%%%%%%

\section{Stationary points}
\label{sec:stationary2}

Recall that Definition~\ref{def:lp} gives a linear program for finding the
direction $\prof'$ that minimises $Df^{\delta}(\prof,\prof')$.
%steepest descent direction, values: $\prof'$ and $V(\prof)$.
%We will use these values to define a stationary point. Recall that $V(\prof) =
%\max_{i \in \mathcal{K}(\prof)} Df^{\delta}_i(\prof, \prof')$, where $\mathcal{K}(\prof)$ is defined in \eqref{eq:k-set}, and by \eqref{eq:dgrad} we
%have that the gradient at $\prof$ is:
%\begin{equation*}
%Df^{\delta}(\prof, \prof') = \max_{i \in \mathcal{K}(\prof)}
%Df^{\delta}_i (\prof, \prof') - f(\prof) = V(\prof) - f(\prof).
%\end{equation*}
Our steepest descent procedure is able to make
progress whenever this gradient is negative, and so a stationary point is any
point $\prof$ for which $Df^{\delta}(\prof, \prof') \ge 0$. In fact, our analysis
requires us to consider $\delta$-stationary points, which we now define.

%Let $\prof$ be the input to the LP, let $\prof'$ and $V(\prof)$ be output of the LP. 
%Recall that $\prof' - \prof$ is the steepest descent direction for
%$f(\prof)$ from point $\prof$, and that $V(\prof)$ is the gradient of this
%direction.
%%
%If the gradient of $f(\prof)$ along direction defined by $\prof'$ is negative, formally 
%$Df(\prof, \prof') = V(\prof) - f(\prof) < 0$, then the value of $f(\prof)$ 
%is decreasing along this direction. 

%When the value of the objective function at point \prof is at least as large as the value
%of the function (that means that we cannot decrease the value of $f(\prof)$ using the steepest 
%descent procedure), then this is a stationary point for $f$.

\begin{definition}[$\delta$-stationary point]
\label{def:stationary}
Let $\prof^*$ be a mixed strategy profile, and let $\delta > 0$.
%Let $V(\prof)$ be the value of the objective function of the steepest descent linear program
%and $\delta > 0$.
We have that $\prof^*$ is a $\delta$-stationary point if for all $\prof' \in
\Delta$:
\begin{equation*}
%\label{eq:stationary}
Df^{\delta}(\prof^*, \prof') \geq - \delta.
\end{equation*}
\end{definition}

%Furthermore, we show that the steepest descent procedure will converge
%to a $\delta$-stationary point or it will reach a point where $f(\prof) \leq 0.5$ after polynomially
%many iterations. 

%%%%%%%%%%%%%%%%%%%%%%%%%%%%%%%%%%
%%%%%%%%%%%%%%%%%%%%%%%%%%%%%%%%%%
%%%%%%%%%%%%%%%%%%%%%%%%%%%%%%%%%%
%%%%%%%%%%%%%%%%%%%%%%%%%%%%%%%%%%
%\section{Approximation properties of $\delta$-stationary points}
%\label{sec:approx}
%
We now show that every $\delta$-stationary point of $f(\prof)$ is a
$(0.5+\delta)$-NE. Recall from Definition~\ref{def:dgrad} that:
\begin{equation*}
Df^{\delta}(\prof, \prof') = \max_{i \in \mathcal{K}(\prof)} Df^{\delta}_i (\prof, \prof') - f(\prof).
\end{equation*}
Therefore, if $\prof^*$ is a $\delta$-stationary point, we must have, for every
direction $\prof'$:
\begin{equation}
\label{eqn:stationary2}
f(\prof^*) \le 
%V(\prof^*) 
\max_{i \in \mathcal{K}(\prof)} Df^{\delta}_i (\prof^*, \prof')
+ \delta.
\end{equation}
Since $f(\prof^*)$ is the maximum regret under the strategy profile $\prof^*$,
in order to show that $\prof^*$ is a $(0.5 + \delta)$-NE, we only have to find
some direction $\prof'$ such that that $\max_{i \in \mathcal{K}(\prof)}
Df^{\delta}_i (\prof, \prof') \le 0.5$. We do this in the following lemma.

%Our bound will be obtained by analysing the values of $Df_i(\prof, \prof')$ at a stationary point $\prof^*$.
%The following lemma shows that in a stationary point $Df_i(\prof^*, \prof') \leq \frac{1}{2}$
%for all $i \in [n]$.
%To begin, we deal with
%the $\max_{k \in \brs} \left \{ \bigl ( \pays{\prof'}{i} \bigr )_k \right \}$
%term from the above equation. The following definition captures the set of
%strategies that can satisfy the conditions on the maximum used in this term.
\begin{lemma}
\label{lem:dfi-bound}
In every stationary point $\prof^*$, there exists a direction $\prof'$ such
that:
\begin{equation*}
\max_{i \in \mathcal{K}(\prof)} Df^{\delta}_i (\prof^*, \prof') \leq 0.5.
\end{equation*}
\end{lemma}

\begin{proof}
%We must show that $Df^{\delta}_i (\prof, \prof') \leq 0.5 $ for all $i \in 
%\mathcal{K}(\prof)$. 
%Let , as defined in \eqref{eq:k-set}, to be the set of players
%that have maximum regret at point \prof.
%Recall that for every point $\prof$, if $\prof'$ is the solution to the steepest
%descent LP, then we have that 
%$V(\prof) = \max_{i \in \mathcal{K}(\prof)} Df^{\delta}_i(\prof,\prof')$. 
First, define $\bar{\prof}$ to be a strategy profile in which 
each player $i \in [n]$ plays a best response against $\prof^*$.
We will set $\prof' = 
\frac{\bar{\prof}+\prof^*}{2}$. 
Then for each $i \in \mathcal{K}(\prof)$, we have that $Df^{\delta}_i(\prof^*,
\prof')$, is less than or equal to:
%when $\prof' = \frac{\bar{\prof}+\prof^*}{2}$
%we have that $Df^{\delta}_i(\prof^*, \frac{\bar{\prof}+\prof^*}{2})$ is equal to:
\begin{align*}
& 
\max_{k \in \brst} \left \{ \bigl ( \pays{\frac{\bar{\prof}+\prof^*}{2}}{i}\bigr )_k \right \} 
- \payt{x^*_i}{\frac{\bar{\prof}+\prof^*}{2}}{i} 
- \payt{\frac{\bar{x_i}+x^*_i}{2}}{\prof^*}{i} + 
\payt{x^*_i}{\prof^*}{i} \\
% LINE 2
 & = 
\frac{1}{2} \cdot \max_{k \in \brst} \left \{ \bigl ( \pays{\bar{\prof}+\prof^*}{i} \bigr )_k \right \} 
- \frac{1}{2} \cdot \payt{x^*_i}{\bar{\prof}}{i} 
- \frac{1}{2} \cdot \payt{\bar{x_i}}{\prof^*}{i} \\
% LINE 3
 & \leq \frac{1}{2} \cdot \left(
\max_{k \in \brst} \left \{ \bigl ( \pays{\bar{\prof}}{i} \bigr )_k \right \} +
\max_{k \in \brst} \left \{ \bigl ( \pays{\prof^*}{i} \bigr )_k \right \} 
- \payt{x^*_i}{\bar{\prof}}{i} 
- \payt{\bar{x_i}}{\prof^*}{i} \right) \\
% LINE 4
 & = \frac{1}{2} \cdot \left(
\max_{k \in \brst} \left \{ \bigl ( \pays{\bar{\prof}}{i} \bigr )_k \right \} 
- \payt{x^*_i}{\bar{\prof}}{i} \right) \quad \text{because $\bar{x_i}$ is a b.r. to $x^*$}\\
% LINE 5
 &  \leq \frac{1}{2} \cdot \max_{k \in \brst} \left \{ \bigl ( \pays{\bar{\prof}}{i} \bigr )_k \right \}  \\
 & \leq \frac{1}{2}.
\end{align*} 
%This mean that for $x' = \frac{\bar{\prof} + \prof^*}{2}$ all $D()
Thus, the point $\prof'$ satisfies
$\max_{i \in \mathcal{K}(\prof)} Df^{\delta}_i (\prof^*, \prof') \leq 0.5$.
\qed
\end{proof}

We can sum up the results of the section  in the following lemma.
\begin{lemma}
\label{lem:main}
Every $\delta$-stationary point $\prof^*$ is a $(0.5 + \delta)$-Nash equilibrium.
\end{lemma}

\section{The time complexity of the algorithm}
\label{sec:converge}

%In order to show that the algorithm converges in a $\delta$-stationary point, 
%we will slightly modify the technique used in \cite{TS}.
%%
%We show that either the algorithm converges to a $\delta$-stationary point, or 
%(assuming that it does not converge) that after $O(\frac{\maxdeg^2}{\delta^2})$ 
%iterations of the steepest descent procedure  it reaches a point \prof that 
%$f(\prof) \leq 0.5$. 
%%
%Roughly, the idea is to find the minimum decrease achieved on the 
%value of $f(\prof)$ in every iteration of steepest descent 
%procedure and use this in order to bound the number of iterations 
%needed to find a point $\prof$ such that $f(\prof) \leq 0.5$. 
%

In this section, we show that Algorithm 1 terminates after a polynomial number
of iterations. Let $\prof$ be a strategy profile that is considered by Algorithm
1, and let $\prof' = Q(\prof)$ be
the solution of the steepest descent LP for $\prof$. These two profiles will be
fixed throughout this section. 
%In this section, we will
%show that if $\prof$ is not a $\delta$-stationary point, then Algorithm 1 will
%makes  progress 
%Recall that Algorithm 1 moves from $\prof$ to the point $\bar{\prof} := \prof + \eps(\prof' -
%\prof)$, where $\epsilon$ is found by solving the optimal distance linear
%program. In this section, we show a lower bound on $f(\prof +\eps(\prof' - \prof)) - f(\prof)$,
 %which corresponds to showing a lower bound on the amount
%of progress made by each iteration of the algorithm.

%To simplify our analysis, we show our lower bound for $\eps =
%\frac{\delta}{\delta + 2}$, which will be fixed throughout the rest of this
%section. Since $\eps = \frac{\delta}{\delta + 2}$ gives a feasible point in the optimal
%distance linear program, a lower bound for the case where
%$\eps = \frac{\delta}{\delta + 2}$ is also a lower bound for the case where
%$\epsilon$ solves the optimal distance LP.

%%%%%%%%%%%%%%%%%%%%%%%%%%%
%%%%%%%%%%%%%%%%%%%%%%%%%%%

We begin by proving a technical lemma that will be crucial for showing our
bound on the number of iterations. To simplify our notation, throughout this
section we define $f_{new} := f(\prof + \eps(\prof' - \prof))$ and $f :=
f(\prof)$. Furthermore, we define $\mathcal{D} = \max_{i \in [n]}
Df^{\delta}_i(\prof, \prof')$. The following lemma, which is proved in
Appendix~\ref{app:gain-bound}, gives a relationship between $f$ and $f_{new}$.

\begin{lemma}
\label{lem:gain-bound}
In every iteration of Algorithm $1$ we have:
\begin{equation}
\label{eq:gain}
f_{new} - f \leq 
\epsilon ( \mathcal{D} - f) + \epsilon^2 (1 - \mathcal{D}).
\end{equation}
\end{lemma}

%%%%%%%%%%%%%%%%%%%%%%%%%%%
%%%%%%%%%%%%%%%%%%%%%%%%%%%
%%%%%%%%%%%%%%%%%%%%%%%%%%%
%%%%%%%%%%%%%%%%%%%%%%%%%%%
%%%%%%%%%%%%%%%%%%%%%%%%%%%
%%%%%%%%%%%%%%%%%%%%%%%%%%%
%%%%%%%%%%%%%%%%%%%%%%%%%%%
%%%%%%%%%%%%%%%%%%%%%%%%%%%
In the next lemma we prove that, if we are not in a $\delta$-stationary point, then
we have a bound on the amount of progress made in each iteration.  We use
this in order to bound the number of iterations needed before we reach a point
\prof where $f(\prof) \leq 0.5 + \delta$.

\begin{lemma}
\label{lem:discount}
Fix $\epsilon = \frac{\delta}{\delta + 2}$, where $0 < \delta \le 0.5$.
Either \prof is a $\delta$-stationary point or:
\begin{equation}
\label{eq:decrease}
f_{new} \leq \left( 1 - \left(\frac{\delta}{\delta + 2} \right)^2 \right) f.
\end{equation}
\end{lemma}

\begin{proof}
Recall that by Lemma \ref{lem:gain-bound} the gain in every iteration of the steepest descent is
\begin{align}
\label{fnew}
f_{new} - f \leq \epsilon(\mathcal{D} - f) + \epsilon^2(1 - \mathcal{D}).
\end{align}
We consider the following two cases:
\begin{enumerate}
\item[a)] $\mathcal{D} - f > -\delta$. Then, by definition, we are in a $\delta$-stationary point.
\item[b)] $\mathcal{D} - f \leq -\delta$.
%
%%%%%%%%%%%%%%%%%%%%%%%%%%%%%%%%%%%%%%%%%%%%
%%%%%%%%%%%%%%%%%%%%%%%%%%%%%%%%%%%%%%%%%%%%
We have set $\eps = \frac{\delta}{\delta + 2}$. If we solve for $\delta$ we get that 
$\delta =  \frac{2\eps}{1-\eps}$. 
Since $\mathcal{D} - f \leq-\delta$, we have that  $(\mathcal{D} - f)(1-\eps) \leq -2\eps$. 
Thus we have:
\begin{align*}
(\mathcal{D} - f)(\eps - 1) & \geq 2\eps \\
0 & \geq (\mathcal{D} - f)(1 - \eps) + 2\eps  \\
0 & \geq (\mathcal{D} - f) + \eps(2 - \mathcal{D} + f) \\
-\eps f - \eps & \geq (\mathcal{D} - f) + \eps(1- \mathcal{D}) \qquad \text{($\eps \geq 0$)}\\
-\eps^2 f - \eps^2 & \geq \eps(\mathcal{D} - f) + \eps^2(1- \mathcal{D}).
\end{align*}
Thus, since $\eps^2 \ge 0$ we get:
\begin{align*}
-\eps^2 f & \geq \eps(\mathcal{D} - f) + \eps^2(1- \mathcal{D})\\
          & \geq f_{new} - f & \text{According to \eqref{fnew}.}
\end{align*}
%%%%%%%%%%%%%%%%%%%%%%%%%%%%%%%%%%%%%%%%%%%%
Thus we have shown that:
\begin{align}
\nonumber
f_{new} - f \leq & -\epsilon^2 f \\
\nonumber
f_{new} \leq & (1-\epsilon^2)f.
\end{align}
Finally, using the fact that $\eps = \frac{\delta}{\delta + 2}$, we get that
\begin{align*}
f_{new} \leq & \left( 1 - \left(\frac{\delta}{\delta + 2} \right)^2 \right) f.
\end{align*}
\end{enumerate}
\qed
\end{proof}
So, when the algorithm has not reached yet a $\delta$-stationary point, there is a decrease
on the value of $f$ that is at least as large as the bound specified in \eqref{eq:decrease} in every iteration of the gradient 
descent procedure.
In the following lemma we prove that after $O(\frac{1}{\delta^2})$ iterations of the 
steepest descent procedure the algorithm finds a point \prof where $f(\prof)\leq 0.5 + \delta$. 
%%%%%%%%%%%%%%%%%%%%%%%%%%%%%
%%%%%%%%%%%%%%%%%%%%%%%%%%%%%%%%%%%
\begin{lemma}
\label{lem:iterations}
After $O(\frac{1}{\delta^2})$ iterations of the steepest descent procedure the
algorithm finds a point \prof where $f(\prof) \leq 0.5 + \delta$.
\end{lemma}
\begin{proof}
%Let $\prof_1$ be the initial strategy profile considered by Algorithm $1$. By
%Lemma~\ref{lem:discount}, we can see that each iteration of Algorithm $1$
%strictly decreases 
%Let $\prof_1$, $\prof_2$, \dots be the strategies considered by the algorithm. 
%Let $\prof$ be a strategy profile
Let $\prof_1$, $\prof_2$, $\dots$, $\prof_k$ be the sequence of strategy profiles that are
considered by 
Algorithm 1. Since the algorithm terminates as soon as it finds a $(0.5 + \delta)$-NE, we have
$f(\prof_i) > 0.5 + \delta$ for every $i < k$. Therefore, for each $i < k$ we
we can apply
Lemma~\ref{lem:main} to argue that $\prof_i$ is not a $\delta$-stationary point,
which then allows us to apply Lemma~\ref{lem:discount} to obtain:
\begin{equation*}
f(\prof_{i+1}) \leq \left( 1 - \left(\frac{\delta}{\delta + 2} \right)^2 \right)
f(\prof_i).
\end{equation*}
So, the amount of progress made by the algorithm in iteration $i$ is:
\begin{align*}
f(\prof_i) - f(\prof_{i+1}) &\ge f(\prof_i) - \left( 1 - \left(\frac{\delta}{\delta + 2} \right)^2 \right)
f(\prof_i) \\
&= \left(\frac{\delta}{\delta + 2} \right)^2 f(\prof_i) \\
&\ge \left(\frac{\delta}{\delta + 2} \right)^2 \cdot 0.5.
\end{align*}
Thus, each iteration of the algorithm decreases the regret by at least
$(\frac{\delta}{\delta + 2})^2 \cdot 0.5$. The algorithm starts at a point
$\prof_1$ with $f(\prof_1) \le 1$, and terminates when it reaches a point
$\prof_k$ with $f(\prof_k) \le 0.5 + \delta$. Thus the total amount of progress
made over all iterations of the algorithm can be at most $1 - (0.5 + \delta)$.
Therefore, the number of iterations used by the algorithm can be at most:
\begin{align*}
\frac{1 - (0.5 + \delta)}{ \left(\frac{\delta}{\delta + 2}\right)^2 \cdot 0.5} 
& \le 
\frac{1 - 0.5}{ \left(\frac{\delta}{\delta + 2}\right)^2 \cdot 0.5}  \\
& = \frac{(\delta + 2)^2}{\delta^2} 
%& = \frac{\delta^2}{\delta^2} + \frac{4\delta}{\delta^2} + \frac{4}{\delta^2}.
 = \frac{\delta^2}{\delta^2} + \frac{4\delta}{\delta^2} + \frac{4}{\delta^2}.
\end{align*}
Since $\delta < 1$, we have that the algorithm terminates after at
most $O(\frac{1}{\delta^2})$ iterations. \qed
\end{proof}
%%%%%%%%%%%%%%%%%%%%%%%%%%%%%%%%%
%%%%%%%%%%%%%%%%%%%%%%%%%%%%%%%%%%%
Lemma~\ref{lem:iterations} implies that that after polynomially many iterations
the algorithm finds a point such that $f(\prof) \leq 0.5 +\delta$, and by
definition such a point is a $(0.5 + \delta)$-NE. Thus we have completed the
proof of Theorem~\ref{thm:main}.

\section{Application: Two-player Bayesian games}

In this section, we define two-player Bayesian games, and show how our algorithm
can be applied in order to efficiently find a $(0.5+\delta)$-Bayesian Nash
equilibrium.
A two-player Bayesian game is played between a \emph{row} player and a
\emph{column} player. Each player has a set of possible \emph{types}, and at the
start of the game, each player is assigned a type by drawing from a known joint
probability distribution. Each player learns his type, but not the type of his
opponent. Our task is to find an approximate Bayesian Nash equilibrium (BNE).

We show that this can be reduced to the problem of finding an $\epsilon$-NE in a
polymatrix game, and therefore our algorithm can be used to efficiently find a
$(0.5+\delta)$-BNE of
a two-player Bayesian game. This section is split into two parts. In the first
part we formally define two-player Bayesian games, and approximate Bayesian Nash
equilibria. In the second part, we give the reduction from two-player Bayesian
games to polymatrix games.

\subsection{Definitions}

\paragraph{\bf Payoff matrices.}

We will use $k_1$ to denote the number of pure strategies of the row player and
$k_2$ to denote the number of pure strategies of the column player. Furthermore,
we will use $m$ to denote the number of types of the row player, and $n$ to
denote the number of types of the column player. 

For each pair of types $i \in [m]$ and $j \in [n]$, there is a $k_1 \times k_2$
bimatrix game $(R,C)_{ij} := (R_{ij}, C_{ij})$ that is played when the row
player has type $i$ and the column player has type $j$. We assume that all
payoffs in every matrix $R_{ij}$ and every matrix $C_{ij}$ lie in the range $[0,
1]$.

\paragraph{\bf Types.}

The distribution over types is specified by a joint probability distribution:
for each pair of types $i \in
[m]$ and $j \in [n]$, the probability that the row player is assigned type $i$ and
the column player is assigned type $j$ is given by $p_{ij}$. Obviously, we
have that: 
\begin{equation*}
\sum_{i=1}^m \sum_{j=1}^n p_{ij} = 1.
\end{equation*}
We also define some useful shorthands:
for all $i \in [m]$ we denote by $p^R_i$ ($p^C_j$) the probability that row (column)
player has type $i \in [m]$ ($j \in [n]$). Formally:
\begin{align*}
p^R_i = \sum_{j=1}^n p_{ij} \qquad \text{for all $i \in [m]$},\\
p^C_j = \sum_{i=1}^m p_{ij} \qquad \text{for all $j \in [n]$}.
\end{align*}
Note that $\sum_{i=1}^m p^R_i = \sum_{j=1}^n p^C_j = 1$.
Furthermore, we denote by $p^R_i(j)$ the conditional probability that 
type $j \in [n]$ will be chosen for column player given that type $i$ is chosen for row player. 
Similarly, we define $p^C_j(i)$ for the column player. Formally:
\begin{align*}
p^R_i(j) = \frac{p_{ij}}{p^R_i} \qquad \text{for all $i \in [m]$}\\
p^C_j(i) = \frac{p_{ij}}{p^C_j} \qquad \text{for all $j \in [n]$}.
\end{align*}
We can see that for given type $t=(i,j)$ we have that 
$p_{ij} = p^R_i \cdot p^R_i(j) = p^C_j \cdot p^C_j(i)$.

\paragraph{\bf Strategies.}
In order to play a Bayesian game, each player must specify a strategy for each
of their types. Thus, a strategy profile is a pair $(\profx, \profy)$, where
$\profx = (x_1, x_2, \dots, x_{m})$ such that each $x_i \in \Delta_{k_1}$, and
where $\profy = (y_1, y_2, \dots, y_{n})$ such that each $y_i \in \Delta_{k_2}$.
This means that, when the row player gets type $i \in [m]$ and the column player
gets type $j \in [n]$, then the game $(R_{ij}, C_{ij})$ will be played, and
the row player will use strategy $x_i$ while the column player will use
strategy~$y_j$.

Given a strategy profile $(\profx, \profy)$, we can define the expected payoff
to both players (recall that the players are not told their opponent's type).
\begin{definition}[Expected payoff]
Given a strategy profile $(\profx, \profy)$ and a type $t=(i,j)$, the expected payoff for the row player is given by:
\begin{align*}
u_R(x_i, \profy) & =   \sum_{j=1}^n p^R_i(j) \cdot x^T_i R_{ij}y_j \\
& =   x^T_i \sum_{j=1}^n p^R_i(j) \cdot R_{ij}y_j 
%\label{eq:rexp}
 %& = x^T_i \sum_{j=1}^n R^*_{ij}y_j.
\end{align*}
Similarly, for the column player the expected payoff is:
\begin{align*}
%\label{eq:cexp}
u_C(\profx, y_j) &=  y^T_j \sum_{i=1}^m p^C_j(i) \cdot C_{ij}^T x_i.  
\end{align*}
\end{definition}

\paragraph{\bf Rescaling.}
Before we define approximate equilibria for two-player Bayesian games, we first
rescale the payoffs. Much like for polymatrix games, rescaling is needed to
ensure that an $\epsilon$-approximate equilibrium has a consistent meaning. Our
rescaling will ensure that, for every possible pair of types, both player's
expected payoff uses the entire range $[0, 1]$. 

For each type $i$ of the row player, we use $\pmax^i_R$ to denote the maximum
expected payoff for the row player when he has type $i$, and we use
$\pmin^i_R$ to denote the minimum expected payoff for the row player when he has
type $i$. Formally, these are defined to be:
\begin{align*}
%\pmax^i_R &=  \max_{a \in [k_1]} \sum_{j=1}^n \max_{b \in [k_2]} \left(R^*_{ij}\right)_{a, b}, \\
%
\pmax^i_R &=  \max_{a \in [k_1]} 
\sum_{j=1}^n 
\max_{b \in [k_2]} \left( p^R_i(j) \cdot R_{ij} \right)_{a, b}, \\
%
%\pmin^i_R &=  \min_{a \in [k_1]} \sum_{j=1}^n \min_{b \in [k_2]}
%\left(R^*_{ij}\right)_{a, b}. \\
%
\pmin^i_R &=  \min_{a \in [k_1]} \sum_{j=1}^n \min_{b \in [k_2]}
\left( p^R_i(j) \cdot R_{ij} \right)_{a, b}.
\end{align*}
Then we apply the transformation 
$T_R^i(\cdot)$ to every element~$z$ of $R_{ij}$, for all types $j$ of the
column player, where: 
%\todo[inline]{according to this we applying scaling to $R$ not $R^*$ - actually
	%that's OK because it creates the game we want to solve, but then we need to
	%compute a new $R'^*$ for the reduction. Maybe just have scaling as first
%thing in the section like an assumption}
\begin{align}
\label{eq:transformR}
T_R^i(z) := \frac{1}{\pmax^i_R-\pmin^i_R} \cdot \left(z - \frac{\pmin^i_R}{n} \right).
\end{align}
%We use $R'_{ij}$ to denote the result of applying this transformation. 
Similarly, we transform all payoff matrices for the column player using
\begin{align}
\label{eq:transformC}
T_C^j(z) := \frac{1}{\pmax^j_C-\pmin^j_C} \cdot \left(z - \frac{\pmin^j_C}{m} \right), 
\end{align}
where $\pmax^j_C$ and $\pmin^j_C$ are defined symmetrically. Note that, after
this transformation has been applied, both player's expected payoffs lie in the
range $[0, 1]$. Moreover, the full range is used: there exists a strategy for
the column player against which one of the row player's strategies has expected
payoff $1$, and there exists a strategy for the column player against which one
of the row player's strategies has expected payoff $0$. From now on we will
assume that the payoff matrices have been rescaled in this way.

%Moreover, there exists a strategy for the
%column player 
%From now one

%then the expected payoffs for the two players are given by 
%$u_R(x_i, \profy) = x^T_i \sum_{j=1}^n R'_{ij}y_j$
%and $u_C(\profx, y_j) =  y^T_j \sum_{j=1}^m C'_{ij}x_i$ respectively.

%Observe that, under this rescaling, we have that $u_R(x_i, \profy) \in [0, 1]$
%for all types $i \in [m]$ of the row player, and all strategies $\profy$ of the
%column player. 
%Note that, by the same reasoning discussed in
%Section~\ref{sec:prelim}, this rescaling can only strengthen our result, as
%every $\epsilon$-BNE in the rescaled game is an $\epsilon$-BNE in the original
%unscaled game.

We can now define approximate Bayesian Nash equilibria for a two-player
Bayesian game.

\begin{definition}[Approximate Bayes Nash Equilibrium ($\eps$-BNE)]
Let $(\profx, \profy)$ be a strategy profile. The profile $(\profx, \profy)$ is
an $\eps$-BNE iff the following conditions hold:
\begin{align}
u_R(x_i, \profy) & \geq u_R(x'_i, \profy) - \eps \quad \text{for all $x'_i \in
\Delta^{k_1}$} \quad \text{for all $i \in [m]$},\\
u_C(\profx, y_j) & \geq u_C(\profx, y'_j) - \eps \quad \text{for all $y'_j \in \Delta^{k_2}$} \quad \text{for all $j \in [n]$}.
\end{align}
\end{definition}

\subsection{The reduction}
In this section we reduce in polynomial time the problem of computing an $\eps$-BNE for a
two-player Bayesian game $\mathcal{B}$ to the problem of computing an $\eps$-NE
of a polymatrix game $\mathcal{P(B)}$. We describe the construction of $\mathcal{P(B)}$  
and prove that every $\eps$-NE for $\mathcal{P(B)}$ maps to an $\eps$-BNE of $\mathcal{B}$.

\paragraph{\bf Construction.}
Let $\mathcal{B}$ be a two-player Bayesian game where the row player has $m$ types and $k_1$ pure strategies
and the column player has $n$ types and $k_2$ pure strategies. 
We will construct a polymatrix game $\mathcal{P(B)}$ as follows.

The game has $m + n$ players. We partition the set of players $[m + n]$ into two
sets: the set $K = \{1, 2, \dots, m\}$ will represent the types of the row
player in $\mathcal{B}$, while the set $L = \{m+1, m+2, \dots, m+n\}$ will
represent the types of the column player in $\mathcal{B}$.
The underlying graph that shows the interactions between the players is a
complete bipartite graph $G = (K \cup L, E)$,  where
every player in $K$ (respectively $L$) plays a bimatrix game with every player
in $L$ (respectively $K$).  The bimatrix game played between vertices $v_i \in
K$ and $v_j \in L$ is defined to be $(R^*_{ij}, C^*_{ij})$, where:
\begin{align}
\label{eq:rstar}
R^*_{ij} &:= p^R_i(j) \cdot R_{ij}\\
\label{eq:cstar}
C^*_{ij} &:= p^C_j(i) \cdot C_{ij}
\end{align}
for all $i \in [m]$ and $j \in [n]$.

%by 
%\todo[inline]{we used to refer to equations (14) and (15) here but they are unscaled}
%rescaled versions of the payoff matrices $R^*_{ij}$ and $C^*_{ij}$ that
%were defined in \eqref{eq:rstar} and \eqref{eq:cstar} with the transformations 
%\eqref{eq:transformR} and \eqref{eq:transformC} applied to them respectively.
%\todo[inline]{change above}

Observe that, for each player $i$ in the $K$, the matrices $R^*_{ij}$ all have
the same number of rows, and for each player $j \in L$, the matrices $C^*_{ij}$
all have the same number of columns. Thus, $\mathcal{P(B)}$ is a valid
polymatrix game.
%Furthermore, all players have expected payoff in $[0, 1]$.
Moreover, we clearly have that $\mathcal{P(B)}$ has the same size as the
original game $\mathcal{B}$. Note that, since we have assumed that the Bayesian
game has been rescaled, we have that for every player in $\mathcal{P(B)}$ 
the minimum (maximum) payoff achievable under pure strategy profiles is $0$
($1$), so no further scaling is needed in order to apply our algorithm.

We can now prove that every $\epsilon$-NE of the polymatrix game is also an
$\epsilon$-BNE of the original two-player Bayesian game, which is the main
result of this section.

\begin{theorem}
%Let $\mathcal{B}$ be a Bayesian game and let $\mathcal{P(B)}$ be the polymatrix 
%game that is constructed as described above. 
Every $\eps$-NE of $\mathcal{P(B)}$ is a $\eps$-BNE for $\mathcal{B}$.
\end{theorem}

\begin{proof}

Let $\profz = (x_1, \ldots, x_m, y_1, \ldots, y_n)$ be an $\eps$-NE for
$\mathcal{P(B)}$. This mean that no player can gain more than \eps by
unilaterally changing his strategy. We define the strategy profile $(\profx,
\profy)$ for $\mathcal{B}$ where $\profx = (x_1, \dots, x_m)$ and $\profy =
(y_1, \dots, y_n)$, and we will show that $(\profx, \profy)$ is an
$\epsilon$-BNE for $\mathcal{B}$.

Let $i \in K$ be a player.
Since, $\profz$ is an $\eps$-NE of $\mathcal{P(B)}$, we have:
\begin{align*}
u_i(x_i, \profz) \geq u_i(x'_i, \profz) - \eps \quad \text{for all $x'_i \in \Delta^{k_1}$}.
\end{align*}
By construction, we can see that player $i$ only interacts with the players from
$L$. Hence his payoff can be written as:
\begin{equation*}
u_i(x_i, \profz) = x^T_i \sum_{j=1}^n R^*_{ij}y_j = u^R(x_i, \profy)
\end{equation*}
and since we are in an $\eps$-NE, we have: 
\begin{align}
\label{eq:proof1}
u^R(x_i, \profy) \geq u^R(x'_i, \profy) - \eps \quad \text{for all $x'_i \in \Delta^{k_1}$}.
\end{align}
This is true for all $i \in K$, thus it is true for all $i \in [m]$.

Similarly, every player $j \in L$ interacts only with players form $K$, thus:
\begin{equation*}
u^C(\profx, y_j) = y^T_j \sum_{i=1}^m (C^*_{ij})^Tx_i.
\end{equation*}
and since we are in an $\eps$-NE we have: 
\begin{align}
\label{eq:proof2}
u^C(\profx, y_j) \geq u_C(\profx, y'_j) - \eps \quad \text{for all $y'_j \in \Delta^{k_2}$}
\end{align} 
and this is true for all $j \in K$, thus it is true for all $j \in [n]$.

Combining now the fact that Equation~\eqref{eq:proof1} is true for all $i \in
[m]$ and that Equation \eqref{eq:proof2} is true for all $j \in [m]$, it is easy
to see that the strategy profile $(\profx, \profy)$ is an $\eps$-BNE for
$\mathcal{B}$.
\qed
\end{proof}

Applying Algorithm 1 to $\mathcal{P(B)}$ thus gives us the following.
\begin{theorem}
A $(0.5+\delta)$-Bayesian Nash equilibrium of a two-player Bayesian game
$\mathcal{B}$ can be found in time polynomial in the input size of $\mathcal{B}$
and $1/\delta$. 
 \end{theorem}

\section{Conclusions and open questions}

We have presented a polynomial-time algorithm that finds a $(0.5+\delta)$-Nash
equilibrium of a polymatrix game for any $\delta>0$.
Though we do not have examples that show that the approximation guarantee is
tight for our algorithm, we do not see an obvious approach to prove a better
guarantee. 
The initial choice of strategy profile affects our algorithm, and it is
conceivable that one may be able to start the algorithm from an efficiently
computable profile with certain properties that allow a better approximation
guarantee.
One natural special case is when there is a constant number of players, which 
may allow one to derive new strategy profiles from a stationary point as done
by Tsaknakis and Sprirakis~\cite{TS}.
It may also be possible to develop new techniques when the number of pure
strategies available to the players is constant, or when the structure of the
graph is restricted in some way.
For example, in the games arising from two-player Bayesian games, the graph is
always bipartite.

%bipartite graphs like in our two-player Bayesian application, 
	%or $k$-partite graphs for a natural extension to ``Bayesian polymatrix
	%games'' with $k$ players}

This paper has considered $\eps$-Nash equilibria, which are the most
well-studied type of approximate equilibria.
However,~$\epsilon$-Nash equilibria have a drawback: since they only require
that the expected payoff is within~$\epsilon$ of a pure best response, it is
possible that a player could be required to place probability
on a strategy that is arbitrarily far from being a best response. 
An alternative, stronger, notion is an 
\emph{$\epsilon$-well supported approximate Nash equilibrium}
($\epsilon$-WSNE).
It requires that players only place probability on strategies that have payoff
within $\epsilon$ of a pure best response. 
Every $\epsilon$-WSNE is an $\epsilon$-Nash, but the converse is not true.
For bimatrix games, the best-known additive approximation
that is achievable in polynomial time gives a
$\bigl(\frac{2}{3}-0.0047\bigr)$-WSNE~\cite{FGSS12}.
It builds on the 
algorithm given by Kontogiannis and Spirakis that achieves a
$\frac{2}{3}$-WSNE in polynomial time~\cite{KS}. 
Recently a polynomial-time algorithm with a better approximation guarantee
have been given for \emph{symmetric} bimatrix games~\cite{CFJ14}.
Note, it has been shown that there is a PTAS for finding $\epsilon$-WSNE of
bimatrix games if and only if
there is a PTAS for $\epsilon$-Nash~\cite{DGP,CDT}.
For $n$-player games with $n>2$ there has been very little work on developing
algorithms for finding $\epsilon$-WSNE. 
This is a very interesting direction, both in general and when $n>2$ is a
constant.

%\section*{Acknowledgements}

\begin{acknowledgements}
We thank Aviad Rubinstein for alerting us to the two-player Baysesian games
application, and Haralampos Tsaknakis for feedback on earlier versions of this
paper.
\end{acknowledgements}

\bibliographystyle{spmpsci}
\bibliography{references}

\newpage 
\appendix
\section{Proof of Lemma~\ref{gradient-lem}}
\label{app:gradient-lem}

Before we begin with the proof, we introduce the following notation.
For a player $i \in [n]$, given a strategy profile $\prof$ and a
subset of $i$'s pure strategies $S \subseteq [m_i]$, we use 
$\maxed{\prof}{S}{i}$  for taking the maximum of the payoffs of $i$ when
the others play according to \prof, and player $i$ is restricted to pick elements
from $S$:
\begin{equation*}
%\label{eq:max-not}
\maxed{\prof}{S}{i} := \max_S \pays{\prof}{i}.
\end{equation*}

In order to find the gradient, we have to calculate the variation of $f_i$ along
the direction $\prof'-\prof$, by evaluating $f(\bar{\prof})$ for points
$\bar{\prof}$ of the form
$$
\bar{\prof} := \prof + \epsilon(\prof' - \prof) = (1 - \eps) \cdot \prof + \epsilon \cdot \prof'.
$$
Recall from~\eqref{eq:regret}, that for $\bar{\prof} \in \mixedstrats$ we have that
$f_i(\bar{\prof}) := \brp{\bar{\prof}}{i} - \pay{\bar{\prof}}{i}$.
In order to rewrite $\brp{\bar{\prof}}{i}$ we introduce notation $\Lambda_i(\prof, \prof', \eps)$ as follows.
\begin{definition}
Given $(\prof, \prof',\eps)$ and $S = \brs$ we define 
$\Lambda_i(\prof, \prof', \eps)$ as:
\begin{equation}
\label{eq:lambda}
\Lambda_i(\prof, \prof', \eps) :=  \max \left\{0, \max_{k \in \bar{S}}\{(\pays{\bar{\prof}}{i})_k\} - \max_{l\in S}\{ (\pays{\bar{\prof}}{i})_l\} \right\}.
\end{equation}
\end{definition}

In the following technical lemma we provide an expression for 
$\brp{\bar{\prof}}{i}$.
%and $\pay{\bar{\prof}}{i}$.
%
%\begin{align*}
%f_i(\prof + \epsilon \cdot (\prof' - \prof)) & = \brp{\prof - \epsilon(\prof' - \prof)}{i} - \pay{\prof - \epsilon(\prof' - \prof)}{i} 
%\end{align*}
In order to rewrite $\brp{\bar{\prof}}{i}$, we use the
following simple observation.
Consider a multiset of numbers  $\{a_1,\ldots,a_n\}$, and the index sets $S \subseteq
[n]$ and $\bar{S} = [n] \setminus S$. We have the following identity:
\begin{align}
\label{eq:split_max}
\max\{ a_1,\ldots,a_n\} \equiv \max_{j\in S}\{ a_j \} + \max \left \{0,\ \max_{k \in
\bar{S}}\{ a_k \} - \max_{j\in S}\{ a_j \} \right\}.
\end{align}
In the following lemma, we use this identity with $S = \brs$ to rewrite
$\brp{\bar{\prof}}{i}$.
\begin{lemma}
\label{om1-label}
Given profiles $\prof$ and $\prof'$ in $\mixedstrats$ and a player $i \in [n]$,
let $S = \brs$.
We have:
\begin{align}
\label{eq:reform-u-star}
\brp{(1 - \eps) \cdot \prof + \epsilon \cdot \prof')}{i}
&  = 
(1 - \epsilon) \cdot \maxed{\prof}{S}{i} 
+
\epsilon \cdot \maxed{\prof'}{S}{i}
+ \Lambda_i(\prof, \prof', \eps).
\end{align}
\end{lemma}
%%%%%%%%%%%%%%%%%%%%%%%%%%%
%%%%%%%%%%%%%%%%%%%%%%%%%%%
\begin{proof}
\begin{align*}
\brp{\bar{\prof}}{i} & =\brp{(1 - \eps) \cdot \prof + \epsilon \cdot \prof')}{i}\\
 & = \max_{k \in [m_i]}\left \{ \bigl ( \pays{\prof + \epsilon (\prof' - \prof)}{i} \bigr )_k \right \}  \hspace*{19mm} \text{By \eqref{eq:brp}}\\
& = \max_{k \in S} \left \{ \bigl ( \pays{\prof + \epsilon (\prof' -
\prof)}{i} \bigr)_k \right\} + \Lambda_i(\prof, \prof', \eps) \qquad \text{By \eqref{eq:split_max}
and \eqref{eq:lambda}}\\ 
& = \max_{k \in S} \left \{ \bigl (  (1 - \epsilon) \cdot \pays{\prof}{i} +  \epsilon \cdot \pays{\prof'}{i} \bigr )_k  \right \} + \Lambda_i(\prof, \prof', \eps).
\end{align*}
Since $S = \brs$, we know that for all $k \in S$ we have that
$(\pays{\prof}{i})_k$ are equal, so we have the following:
\begin{align*}
\max_{k \in S} \left \{ \bigl (  (1 - \epsilon) \cdot \pays{\prof}{i} +  \epsilon \cdot \pays{\prof'}{i} \bigr )_k  \right \} & =
 \max_{k \in S} \left \{ \bigl (  (1 - \epsilon) \cdot \pays{\prof}{i} \bigr)_k \right \} + 
\max_{k \in S} \left \{ \bigl ( \epsilon \cdot \pays{\prof'}{i} \bigr )_k  \right \} \\
 & = (1-\eps) \cdot \maxed{\prof}{S}{i} + \eps \cdot \maxed{\prof'}{S}{i}
\end{align*}
and we get the claimed result.
\qed
\end{proof}

%\end{proof}
%%%%%%%%%%%%%%%%%%%%%%%%%%%%%%%%%%%%%%%%%%%%%%%%%
We will use the expression~\eqref{eq:reform-u-star} for $\brp{\bar{\prof}}{i}$,
along with the following reformulation of $\pay{\bar{\prof}}{i}$:
\begin{align}
\nonumber
\pay{\bar{\prof}}{i} & = \pay{\prof + \epsilon(\prof' - \prof)}{i}\\
\nonumber
 & = \payt{x_i + \epsilon(x'_i - x_i)}{\prof + \epsilon(\prof' - \prof)}{i}\\
\nonumber
 & = \payt{x_i}{\prof}{i} + \eps\cdot\payt{x_i}{\prof' - \prof}{i} +
 	\eps\cdot\payt{x'_i - x_i}{\prof}{i} +  \eps^2\cdot\payt{x'_i - x_i}{\prof' - \prof}{i} \\
\nonumber
 & = \pay{\prof}{i} + \eps\cdot\payt{x_i}{\prof'}{i} -
 	\eps\cdot\payt{x_i}{\prof}{i} + \eps\cdot\payt{x'_i}{\prof}{i} + \eps\cdot\payt{x_i}{\prof}{i} 
 	-  \eps^2\cdot\pay{\prof' - \prof}{i} \\
\label{eqn:fromA}
 & = (1-\eps) \cdot \pay{\prof}{i} + \epsilon \bigl ( \payt{x_i}{\prof'}{i} + \payt{x'_i}{\prof}{i} - \pay{\prof}{i} \bigr ) + 
		\eps^2 \cdot \pay{\prof'-\prof}{i}.
\end{align}

 %& = \pay{\prof}{i} + \epsilon \bigl ( \payt{x_i}{\prof'}{i} + \payt{x'_i}{\prof}{i} - 2\pay{\prof}{i} \bigr ) + 
%\eps^2 \pay{\prof'-\prof}{i} \\
 %& = (1-\eps) \pay{\prof}{i} + \epsilon \bigl ( \payt{x_i}{\prof'}{i} + \payt{x'_i}{\prof}{i} - \pay{\prof}{i} \bigr ) + 
%\eps^2 \pay{\prof'-\prof}{i}.
%\end{align*}
%
%\subsubsection{Combining the two parts together.}
%\begin{proof}[Lemma \ref{gradient-lem}]
We now use these reformulations to prove the following lemma.

\begin{lemma}
\label{lem:help-2}
We have that $f_i(\bar{\prof}) - f(\prof)$ is equal to:
\begin{equation*}
\epsilon \big( Df_i(\prof,\prof') - f(\prof) \big) + \Lambda_i(\prof, \prof', \eps) - \epsilon^2 \pay{\prof' - \prof}{i}
 - (1-\epsilon) \max_{j \in [n]} \bigl\{ f_j(\prof) - f_i(\prof) \bigr\}.
\end{equation*}
\end{lemma}
\begin{proof}
Recall that $S = \brs$.
For a given $i \in [n]$, using Lemma \ref{om1-label} and the reformulation for $u_i(\bar{\prof})$, we have:
\begin{align*}
f_i(\bar{\prof}) - f(\prof) & =  \brp{\bar{\prof}}{i} - \pay{\bar{\prof}}{i} - f(\prof)\\
 & = (1 - \epsilon) \cdot \maxed{\prof}{S}{i} 
+
\epsilon \cdot \maxed{\prof'}{S}{i}
+ \Lambda_i(\prof, \prof', \eps) \\
& - (1-\eps) \pay{\prof}{i} + \epsilon \bigl (-\payt{x_i}{\prof'}{i} - \payt{x'_i}{\prof}{i} + \pay{\prof}{i} \bigr ) - 
\eps^2 \pay{\prof'-\prof}{i}
- f(\prof).
\end{align*}
Recall from \eqref{eq:regret} that $f_i(\prof) = \maxed{\prof}{S}{i} - \pay{\prof}{i}$, 
so the formula above is equal to:
$$\epsilon \bigl ( \maxed{\prof'}{S}{i} - \payt{x_i}{\prof'}{i} - \payt{x'_i}{\prof}{i} + \pay{\prof}{i}  \bigr ) + \Lambda_i(\prof, \prof', \eps) - \epsilon^2 \pay{\prof'-\prof}{i}
+  (1-\epsilon)f_i(\prof) - f(\prof).$$
Using now \eqref{eq:Df_i} for $Df_i(\prof, \prof')$, the above formula becomes:
\begin{align}
\nonumber
 & \epsilon \cdot Df_i(\prof,\prof') + \Lambda_i(\prof, \prof', \eps) - \epsilon^2 \pay{\prof' - \prof}{i}
+  (1-\epsilon)f_i(\prof) - f(\prof) = \\
\nonumber
 & \epsilon \cdot Df_i(\prof,\prof') + \Lambda_i(\prof, \prof', \eps) - \epsilon^2 \pay{\prof' - \prof}{i}
+  (1-\epsilon)f_i(\prof) - (1-\eps)f(\prof) -\eps f(\prof) = \\
\nonumber
 & \epsilon \big( Df_i(\prof,\prof') - f(\prof) \big) + \Lambda_i(\prof, \prof', \eps) - \epsilon^2 \pay{\prof' - \prof}{i} - (1-\epsilon) \big( f(\prof) - f_i(\prof) \big).
\end{align}
Recall now that $f(\prof) = \max_{j \in [n]} f_j(\prof)$. Thus the term
$f(\prof) - f_i(\prof)$ can be written as 
$\max_{j \in [n]} \bigl\{ f_j(\prof) - f_i(\prof) \bigr\}$. So, the expression above is equivalent to
\begin{equation*}
\epsilon \big( Df_i(\prof,\prof') - f(\prof) \big) + \Lambda_i(\prof, \prof', \eps) - \epsilon^2 \pay{\prof' - \prof}{i} 
- (1-\epsilon) \max_{j \in [n]} \bigl\{ f_j(\prof) - f_i(\prof) \bigr\}.
\end{equation*}
\qed
\end{proof}
Now we are ready to prove Lemma~\ref{gradient-lem}.
Recall from definition \ref{def:gradient} for the gradient that 
\begin{align}
\nonumber
Df(\prof,\prof') & = \lim_{\eps \rightarrow 0}\frac{1}{\eps} \bigl( f((1 - \eps) \cdot \prof + \epsilon \cdot \prof') - f(\prof) \bigr) \\
\nonumber
 & = \lim_{\eps \rightarrow 0}\frac{1}{\eps} \bigl( f(\bar{\prof}) - f(\prof) \bigr) \\
\nonumber
 & = \lim_{\eps \rightarrow 0}\frac{1}{\eps} \left( \max_{i \in [n]}f_i(\bar{\prof}) - f(\prof) \right)\\
\nonumber
& = \lim_{\eps \rightarrow 0}\frac{1}{\eps} \left( \max_{i \in [n]} \bigl( f_i(\bar{\prof}) - f(\prof) \bigr) \right) \\
\label{eq:help-3}
& =  \max_{i \in [n]} \left( \lim_{\eps \rightarrow 0}\frac{1}{\eps}  \bigl( f_i(\bar{\prof}) - f(\prof) \bigr) \right).
\end{align}
We will now use lemma \ref{lem:help-2} to study the limit 
$\lim_{\eps \rightarrow 0} ( f_i(\bar{\prof}) - f(\prof) \bigr)$ for all $i \in [n]$.
Firstly, we deal with $\Lambda(\prof, \prof', \eps)$.
It is easy to see that $\lim_{\eps \rightarrow 0} \bigl( \prof + \eps(\prof' - \prof) \bigr) = \prof$.
Then, when $S = \brs$ we have that 
$$\lim_{\eps \rightarrow 0} \left( \max_{k \in \bar{S}}\{(\pays{\bar{\prof}}{i})_k\} - \max_{l\in S}\{ (\pays{\bar{\prof}}{i})_l\} \right) < 0.$$
This is true from the definition of pure best response strategies.
So, from equation \eqref{eq:lambda} for $\Lambda_i(\prof, \prof', \eps)$ it is true that 
$\lim_{\epsilon \rightarrow 0}\Lambda_i(\prof, \prof', \eps) = 0$. 

Furthermore, the term $\epsilon^2 \cdot \pay{\prof' - \prof}{i}$ 
when is divided by $\eps$ equals to $\epsilon \cdot \pay{\prof' - \prof}{i}$, thus $\lim_{\eps \rightarrow 0} \bigl(\epsilon \cdot \pay{\prof' - \prof}{i} \bigr) = 0$.

Moreover, the term
$$
\lim_{\epsilon \rightarrow 0} \left(-\frac{1-\eps}{\eps} \cdot \max_{j \in [n]} \bigl\{ f_j(\prof) - f_i(\prof) \bigr\}\right)
$$ 
is either 0 when $f_i(\prof) = f(\prof)$, i.e player $i$ has the maximum regret and $\max_{j \in [n]} \bigl\{ f_j(\prof) - f_i(\prof) \bigr\} = 0$, or $-\infty$ otherwise, because $\max_{j \in [n]} \bigl\{ f_j(\prof) - f_i(\prof) \bigr\} > 0$.

To sum up, if $f_i(\prof)$ achieves the maximum regret at point $\prof'$, then the limit $\lim_{\epsilon \rightarrow 0} \bigl( f_i(\bar{\prof}) - f(\prof)\bigr) = Df_i(\prof, \prof') - f(\prof)$, otherwise the limit equals $-\infty$. 

From \eqref{eq:help-3} for the gradient we want the maximum of these quantities, thus we have the claimed result.
%\qed
%\end{proof}

\section{Proof of Lemma~\ref{lem:gain-bound}}
\label{app:gain-bound}

%
%Our goal is to prove a bound on how much progress is made in each iteration of
%Algorithm 1.
%In order to prove Lemma~\ref{lem:gain-bound} we will find how much we gain in
%every iteration of the steepest descent procedure when $\eps$ is fixed to 
%$\frac{\delta}{\delta+2}$. The gain is how much the maximum regret is decreased in an iteration. Formally, 
Throughout this proof, $\prof, \prof', \bar{\prof}$, and $\epsilon$ will be
fixed as they are defined in Section~\ref{sec:converge}. In order to prove this
lemma, we must show a bound on:
\begin{equation*}
f(\bar{\prof}) - f(\prof) = \max_{i \in [n]}f_i(\bar{\prof}) - f(\prof).
\end{equation*}

Before we start the analysis we need to redefine the term
$\Lambda^{\delta}_i(\prof, \prof', \eps)$ in order to prove an analogous version
of Lemma \ref{om1-label} when $\delta$-best responses are used.
\begin{definition}
We define $\Lambda^{\delta}_i(\prof, \prof', \eps)$ as:
\begin{equation}
\label{eq:delta-lambda}
\Lambda^{\delta}_i(\prof, \prof', \eps) 
 :=  \max \left\{0, \max_{k \in \overline{\brd}}\{(\pays{\bar{\prof}}{i})_k\} - \max_{l\in \brd}\{ (\pays{\bar{\prof}}{i})_l\} \right\}.
\end{equation}
\end{definition}
We now use this definition to prove the following lemma.

\begin{lemma}
\label{lem:equival}
We have:
\begin{align}
\label{eq:lambda-u-star}
\brp{(1 - \eps) \cdot \prof + \epsilon \cdot \prof')}{i}
\leq (1 - \eps)\max_{k \in \brd} \bigl(\pays{\prof}{i})_k + \eps \max_{k \in \brd} (\pays{\prof'}{i}\bigr)_k  + \Lambda^{\delta}_i(\prof, \prof', \eps).
\end{align}
\end{lemma}
\begin{proof}
We have:
\begin{align*}
\brp{(1 - \eps) \cdot \prof + \epsilon \cdot \prof')}{i} & =
\max_{k \in [m_i]} \bigl(\pays{(1 - \eps) \cdot \prof + \epsilon \cdot \prof'}{i}\bigr)_k \\
& = \max_{k \in \brd} \bigl(\pays{(1 - \eps) \cdot \prof + \epsilon \cdot \prof'}{i}\bigr)_k  + \Lambda^{\delta}_i(\prof, \prof', \eps) \quad \text{Using \eqref{eq:split_max}}\\ 
& \leq (1 - \eps)\max_{k \in \brd} \bigl(\pays{\prof}{i}\bigr)_k + \eps \max_{k \in \brd} \bigl(\pays{\prof'}{i}\bigr)_k  + \Lambda^{\delta}_i(\prof, \prof', \eps).
\end{align*}
\qed
\end{proof}
We will use the reformulation from Equation~\eqref{eqn:fromA} for
$\pay{\bar{\prof}}{i}$: 
\begin{align}
\label{eq:reform2}
\pay{\bar{\prof}}{i} = (1-\eps) \cdot \pay{\prof}{i} + \epsilon \bigl ( \payt{x_i}{\prof'}{i} + \payt{x'_i}{\prof}{i} - \pay{\prof}{i} \bigr ) + 
		\eps^2 \cdot \pay{\prof'-\prof}{i}.
\end{align}
The correctness of this was proved in Appendix \ref{app:gradient-lem}.
Now we use all the these reformulations in order to prove the following lemma.

\begin{lemma}
\label{lem:dgain}
We have that $f_i(\bar{\prof}) - f(\prof)$ is less than or equal to:
\begin{equation}
\label{eq:dgain1}
\epsilon \bigl( Df^{\delta}_i(\prof,\prof') - f(\prof) \bigr) + 
\Lambda^{\delta}_i(\prof, \prof', \eps) - \epsilon^2 \pay{\prof' - \prof}{i} - (1-\epsilon) \max_{j \in [n]} \left\{ f_j - f_i\right\}.
\end{equation}
\end{lemma}
\begin{proof}
Recall that, by definition, we have that:
\begin{equation*}
f_i(\bar{\prof}) = \brp{\bar{\prof}}{i} - \pay{\bar{\prof}}{i}.
\end{equation*}
Thus, we can apply Lemma~\ref{lem:equival} along with the reformulation given in
Equation~\eqref{eq:reform2} for \pay{\bar{\prof}}{i} to prove that
$f_i(\bar{\prof}) - f(\prof)$ is less than or equal to:
\begin{align*}
(1 - \eps) & \max_{k \in \brd} \bigl(\pays{\prof}{i})_k + \eps \max_{k \in \brd} (\pays{\prof'}{i}\bigr)_k  + \Lambda^{\delta}_i(\prof, \prof', \eps) \\
 & - (1-\eps) \pay{\prof}{i} + \epsilon \bigl (-\payt{x_i}{\prof'}{i} - \payt{x'_i}{\prof}{i} + \pay{\prof}{i} \bigr ) - 
\eps^2 \pay{\prof'-\prof}{i}
- f(\prof).
\end{align*}
We can now use the fact that $\max_{k \in
\brd} \bigl(\pays{\prof}{i}\bigr)_k - \pay{\prof}{i} = f_i(\prof)$ and
the definition of 
$Df^{\delta}_i(\prof, \prof')$ given in \eqref{eq:delta-dfi}
%, to obtain
%using similar analysis as in Appendix~\ref{app:gradient-lem}, 
to prove that the expression above is equivalent to:
\begin{align}
\nonumber
& \epsilon \cdot Df^{\delta}_i(\prof,\prof') + \Lambda^{\delta}_i(\prof, \prof', \eps) - \epsilon^2 \pay{\prof' - \prof}{i}
+  (1-\epsilon)f_i(\prof) - f(\prof) \\
\nonumber
 &= \epsilon \cdot Df^{\delta}_i(\prof,\prof') + \Lambda^{\delta}_i(\prof, \prof', \eps) - \epsilon^2 \pay{\prof' - \prof}{i}
+  (1-\epsilon)f_i(\prof) - (1-\eps)f(\prof) -\eps f(\prof)  \\
\nonumber
 &= \epsilon \big( Df^{\delta}_i(\prof,\prof') - f(\prof) \big) + \Lambda^{\delta}_i(\prof, \prof', \eps) - \epsilon^2 \pay{\prof' - \prof}{i} - (1-\epsilon) \big( f(\prof) - f_i(\prof) \big)  \\
\nonumber
&= \epsilon \big( Df^{\delta}_i(\prof,\prof') - f(\prof) \big) + \Lambda^{\delta}_i(\prof, \prof', \eps) - \epsilon^2 \pay{\prof' - \prof}{i} - 
(1-\epsilon) \max_{j \in [n]} \bigl\{ f_j(\prof) - f_i(\prof) \bigr\}.
\end{align}
This completes the proof.
\qed
\end{proof}

Having shown Lemma~\ref{lem:dgain}, we will now study each term of
\eqref{eq:dgain1} and provide bounds for each of them.
To begin with, it is easy to see that for all $i \in [n]$ we have that 
$\max_{j \in [n]} \bigl\{ f_j(\prof) - f_i(\prof) \bigr\} \geq 0$, and since
$\eps < 1$, we have that 
$(1-\epsilon) \max_{j \in [n]} \bigl\{ f_j(\prof) - f_i(\prof) \bigr\} \geq
0$. Thus, Equation \eqref{eq:dgain1} is less than or equal to:
\begin{equation}
\label{eq:dgain2}
\epsilon \big( Df^{\delta}_i(\prof,\prof') - f(\prof) \big) + \Lambda^{\delta}_i(\prof, \prof', \eps) - \epsilon^2 \pay{\prof' - \prof}{i}.
\end{equation} 

Next we consider the term $\Lambda^{\delta}_i(\prof, \prof', \eps)$. In the
following technical lemma we prove that $\Lambda^{\delta}_i(\prof, \prof', \eps)
= 0$ for all $i \in [n]$. 
%
%We will use $\epsilon = \frac{\delta}{\delta + \maxdeg + 1}$ as the step for each iteration 
%of gradient descent procedure, in order to prove that our algorithm needs 
%$O(\frac{\maxdeg^2}{\delta^2})$ iterations to find a $(0.5 + \delta)$-NE.

\begin{lemma}
\label{lem:d-lambda}
%For any $\delta > 0$ and $\epsilon = \frac{\delta}{\delta + 2}$, during every iteration of 
%Algorithm~1, 
We have
$\Lambda^{\delta}_i(\prof, \prof', \eps) = 0$ for all $i \in [n]$.
\end{lemma}
\begin{proof}
%Let $\prof, \prof'$ be the (fixed) strategy profiles in a step of the steepest descent.
%
%If now we take $\eps$ as a
%parameter of $\Lambda^{\delta}_i(\prof, \prof', \eps)$ then 
According to equation \eqref{eq:delta-lambda} for $\Lambda^{\delta}_i(\prof,
\prof', \eps)$, we have: 
\begin{align*}
%\label{eq:lambda-delta2}
\Lambda^{\delta}_i(\prof, \prof', \eps) 
& =  \max \left\{0, \max_{k \in \overline{\brd}}\{(\pays{\bar{\prof}}{i})_k\} - \max_{l\in \brd}\{ (\pays{\bar{\prof}}{i})_l\} \right\}.
\end{align*}
We can rewrite this expression as follows. First define:
\begin{equation*}
Z(\prof, \prof', \epsilon, k) = 
(\pays{\bar{\prof}}{i})_k - \max_{l\in \brd}\{ (\pays{\bar{\prof}}{i})_l\}.
\end{equation*}
Then we have:
\begin{equation*}
\label{eq:dlzero}
\Lambda^{\delta}_i(\prof, \prof', \eps) = \max\bigg\{0, \max_{k \in \overline{\brd}}
\Bigl\{Z(\prof, \prof', \epsilon, k)\Bigr\}\bigg\}.
\end{equation*}
Our goal is to show that, for our chosen value of $\epsilon$, we have 
$\Lambda^{\delta}_i(\prof, \prof', \eps) = 0$. For this to be the case, we
must have that $Z(\prof, \prof', \epsilon, k) \le 0$ for all $k \in
\overline{\brd}$. In the rest of this proof, we will show that this is indeed
the case.

By definition, we have that:
\begin{align}
\label{eq:new1}
(\pays{\bar{\prof}}{i})_k = 
 \bigl( \pays{\prof}{i} + \eps (\pays{\prof'}{i} - \pays{\prof}{i}) \bigr)_k.
\end{align}
The term $\max_{l\in \brd}\{ (\pays{\bar{\prof}}{i})_l\}$ can be written as
follows:
\begin{align}
\nonumber
  \max_{l\in \brd}  \{ (\pays{(1-\eps)\prof + \eps\prof'}{i})_l\} 
 & \geq \max_{l\in \brd}  \{ (\pays{(1-\eps)\prof}{i})_l\} \\
\nonumber
 & = (1-\eps) \cdot \max_{l\in \brd}  \{ (\pays{\prof}{i})_l\} \\
\label{eq:new2}
 & = \max_{l\in \brd}\{ (\pays{\prof }{i})_l\} -
\eps \cdot \max_{l\in \brd} \{ (\pays{\prof}{i})_l\}.
\end{align}
We now substitute these two bounds into the definition of $Z(\prof, \prof', \eps,
k)$.
%Let us denote by $\ones$ an $m_i$-dimensional vector with 1 in every entry.
We have:
\begin{multline}
\label{eq:zbound}
Z(\prof, \prof', \epsilon, k) \le
\pays{\prof}{i}_k - \max_{l\in \brd}\{ (\pays{\prof }{i})_l\}
+ %\\
\eps \bigg( 
\pays{\prof'}{i}_k - \pays{\prof}{i}_k 
% - \ones \max_{l\in \brd}\{ (\pays{\prof'}{i})_l\} 
+ \max_{l\in \brd} \{ (\pays{\prof}{i})_l\}
\bigg).
\end{multline}
From the definition of $\delta$-best responses (Definition
\ref{def:dbr}), we know that for all $k \in
\overline{\brd}$:
\begin{equation*}
\pays{\prof}{i}_k - \max_{l\in \brd}\{ (\pays{\prof }{i})_l\}
 < - \delta.
\end{equation*}
Furthermore, since we know that the maximum payoff for player $i \in [n]$ 
is 1, we have the following trivial bound for all $k \in \overline{\brd}$:
\begin{equation*}
\pays{\prof'}{i}_k - \pays{\prof}{i}_k 
%- \ones \max_{l\in \brd}\{ (\pays{\prof'}{i})_l\} 
+ \max_{l\in \brd} \{ (\pays{\prof}{i})_l\} \le 2.
\end{equation*}
Substituting these two bounds into Equation~\eqref{eq:zbound} gives, for all $k
\in \overline{\brd}$:
\begin{equation*}
Z(\prof, \prof', \epsilon, k) \le - \delta + \epsilon \cdot 2.
\end{equation*}
Thus, for each $k \in \overline{\brd}$, we have that $Z(\prof, \prof',
\epsilon, k) \le 0$ whenever:
\begin{equation*}
- \delta + \epsilon \cdot 2 \le 0,
\end{equation*}
and this is equivalent to:
\begin{equation*}
\eps \le \frac{\delta}{2}.
\end{equation*}
This inequality holds by the definition of $\epsilon$, so we have $Z(\prof,
\prof', \eps, k) \le 0$ for all $k \in \overline{\brd}$, which then implies that
$\Lambda^{\delta}_i(\prof, \prof', \eps) \le 0$. \qed
%\begin{align}
%\eps' \leq
%\min_{k \in \overline{\brd}} \left( 
%\frac{
%\max_{l\in \brd}\{ (\pays{\prof }{i})_l\} - \pays{\prof}{i}
%}{
%\pays{\prof'}{i} - \pays{\prof}{i} - \max_{l\in \brd}\{ (\pays{\prof'}{i})_l\} + \max_{l\in \brd} \{ (\pays{\prof}{i})_l\}
%}
%\right)_k.
%\end{align}
%This means that for every 
%$\eps \leq \frac{\degree \delta}{2\degree} = \frac{\delta}{2}$ the term
%$\Lambda^{\delta}_i(\prof, \prof', \eps)$ is zero. Thus it is zero for
%$\eps = \frac{\delta}{\delta + 2}$ too. We can see that the analysis
 %does not depend on $i$, so it is true for all $i \in [n]$.
\end{proof}
%%%%%%%%%%%%%%%%%%%%%%%%%%%%%%%%%%%%%%%%%%%%

%Next, using simple calculations we derive that, for all $i\in [n]$, 
%$Df^{\delta}_i(\prof, \prof') - 1$ is a lower bound for $\frac{1}{\degree} \pay{\prof' - \prof}{i}$.
%%
%In the following lemma we prove this formally.
Next we consider the term $\pay{\prof' - \prof}{i}$ in
Equation~\eqref{eq:dgain2}. The following lemma provides a simple lower bound
for this term.
\begin{lemma}
\label{lem:u2-bound}
For all $i \in [n]$, we have $Df^{\delta}_i(\prof, \prof') - 1 \leq
\pay{\prof' - \prof}{i}$.
\end{lemma}
\begin{proof}
For $\pay{\prof' - \prof}{i}$ we have the following:
\begin{align}
\nonumber
\pay{\prof' - \prof}{i} & = \payt{x'_i - x_i}{\prof' - \prof}{i} \\
\nonumber
& = \payt{x'_i}{\prof' - \prof}{i} - \payt{x_i}{\prof' - \prof}{i} \\
\label{eq:u2-bound-1}
& = \payt{x'_i}{\prof'}{i} - \payt{x'_i}{\prof}{i} - \payt{x_i}{\prof'}{i} + \payt{x_i}{\prof}{i}.
\end{align}
Recall from \eqref{eq:delta-dfi} that
$$ Df^{\delta}_i(\prof, \prof') = \max_{k \in \brd} \left \{ \bigl ( \pays{\prof'}{i}
\bigr )_k \right \} - \payt{x_i}{\prof'}{i} - \payt{x'_i}{\prof}{i}+ \payt{x_i}{\prof}{i}. $$
We can see that \eqref{eq:u2-bound-1} and \eqref{eq:delta-dfi} differ only in terms
$\payt{x'_i}{\prof'}{i}$ and $\max_{k \in \brd} \left \{ \bigl ( \pays{\prof'}{i}\bigr )_k \right \}$
respectively. 
We know that $\max_{k \in \brd} \left \{ \bigl ( \pays{\prof'}{i}\bigr )_k \right \} \leq 1$.
Then, we can see that $Df^{\delta}_i(\prof, \prof') - 1 \leq \pay{\prof' - \prof}{i}$.
\qed
\end{proof}
Recall that $\mathcal{D} = \max_{i \in [n]} Df^{\delta}_i(\prof,\prof')$ and
$f_{new} = f(\bar{\prof})$ and $f = f(\prof)$.
We can now apply the bounds from Lemma~\ref{lem:d-lambda} and
Lemma~\ref{lem:u2-bound} to Equation~\eqref{eq:dgain2} to obtain:
\begin{align*}
\label{eq:decrease-long}
f_{new} - f & \leq \max_{i \in [n]} \left\{ 
\epsilon \bigl( Df^{\delta}_i(\prof,\prof') - f(\prof) \bigr) -
\eps^2 \bigl( Df^{\delta}_i(\prof, \prof') - 1 \bigr)
\right\}\\
& \leq \max_{i \in [n]} \left\{ 
\epsilon \bigl( Df^{\delta}_i(\prof,\prof') - f(\prof) \bigr) -
\eps^2 \bigl( \mathcal{D} - 1 \bigr)
\right\}\\
 & =  \eps (\mathcal{D} - f) + \eps^2(1 - \mathcal{D}).
\end{align*}
This completes the proof of Lemma~\ref{lem:gain-bound}.

\end{document}